%% file: main.tex
\documentclass[aps,prl,reprint,superscriptaddress,longbibliography]{revtex4-2}
\usepackage[T1]{fontenc}
\usepackage{lmodern}
\usepackage{amsmath,amssymb,amsthm,graphicx,booktabs}
\usepackage[colorlinks=true,linkcolor=blue,citecolor=blue,urlcolor=blue]{hyperref}
\input{macros}

\newtheorem{theorem}{Theorem}
\usepackage{mathrsfs}
\newtheorem{proposition}{Proposition}
\begin{document}
\title{Finite Mass Response to Vanishing Perturbations at Cauchy Horizons}
\author{Yi-kun Li}
\email{liyikun@xao.ac.cn}
\affiliation{State Key Laboratory of Radio Astronomy and Technology,
Xinjiang Astronomical Observatory, CAS, 150 Science 1-Street,
Urumqi 830011, China}
\affiliation{College of Astronomy and Space Science,
University of Chinese Academy of Sciences, No.1 Yanqihu East Road,
Beijing 101408, China}
\date{October 19, 2026}
\begin{abstract}
Smooth transverse perturbations with vanishing $L^p$ norm, for fixed
finite $p\ge4$, can leave a finite positive mass increment at a Cauchy
horizon. We prove this at a critical surface-gravity relation in the
spherical Einstein--Maxwell--scalar system with positive cosmological
constant. Each solution has bounded mass and local finite-energy weak
extensions. The response escapes every fixed time window. An analytic
law determines its delay and duration, and resolves enhancement or
cancellation for overlapping packets.
\end{abstract}
\maketitle

Black-hole inner horizons test how much an initial perturbation can
matter after an arbitrarily long evolution. A Cauchy horizon bounds
the region determined by initial data. Its fate under perturbations
therefore bears on both the persistence of the black-hole interior and
the predictive content of the Einstein equations. The familiar possibility
is mass inflation: radiation is amplified until the internal mass
diverges \cite{PoissonIsrael1990}. For charged black holes with positive
cosmological constant, exponential decay competes with this amplification,
and the nonlinear equations also admit bounded-mass interiors with weak
extensions
\cite{Costa2018,Rossetti2025}. These outcomes depend on the decay and
regularity of the initial data. Rough transverse data can generate
mass inflation even when smoother data admit a more regular horizon
\cite{RossettiRough}.

There is a further question within the bounded-mass regime.
Does making the initial perturbation small also make the final
interior mass response small? Finite-time continuous dependence does
not settle this, because the Cauchy horizon is reached at an infinite
value of a characteristic time coordinate. Energy can move to later
times as the initial norm decreases.
For linear waves on a fixed Reissner--Nordstr\"om--de Sitter background,
time translation and nonzero transmission already yield energy
amplification at the Cauchy horizon \cite[Sec.~5]{DafermosSR2018}.

The nonlinear problem adds a physical requirement: the amplified field
must propagate through the geometry that it produces. Its energy changes
the mass and the null expansion, which in turn affect transmission.
To compare terminal responses, the solutions must reach a common
portion of the horizon, and the radiation must be controlled throughout
the evolution. An estimate on any fixed time interval leaves open the
amount of mass deposited later.

We determine a finite mass response in fully nonlinear solutions.
A sequence of smooth packets has vanishing initial $L^p$ norm and
converging initial geometry. Each packet leaves a finite positive mass
increment at the Cauchy horizon. The increment has the same limiting size as the packet
is pushed toward the initial corner, but its appearance is delayed.
The limiting response also has a definite duration, independently
controlled by the packet's logarithmic width.
The response follows from a positive radial mode and an energy estimate
on the evolving geometry.

The key balance occurs when the surface gravities select an energy
that survives concentration of the initial data. Spreading the packet
over a growing range of logarithmic scales lowers its initial norm
while retaining this energy. The mass law shows how a horizon can
remain weakly extendible while
its terminal mass fails to depend continuously on the initial data in
the chosen norm.

\begin{figure*}[t]
\centering
\includegraphics[width=170mm]{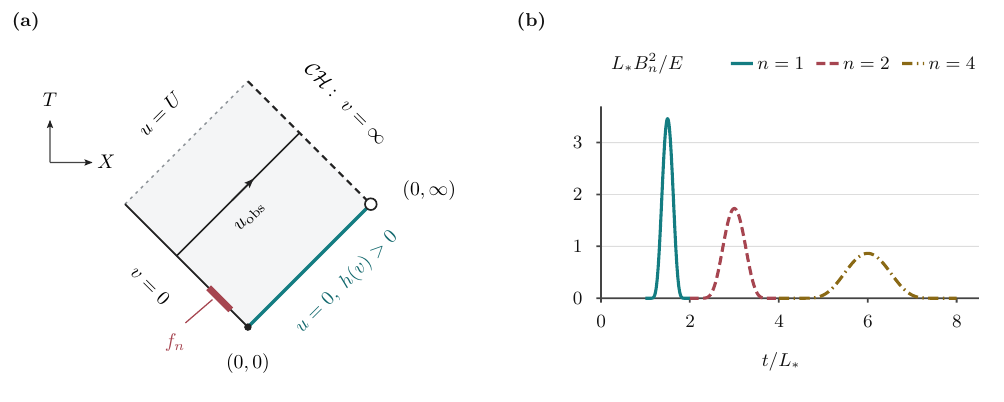}
\caption{\label{fig:spacetime}Initial concentration with a shrinking
$L^p$ norm and fixed envelope energy.
(a) Local characteristic region in a Penrose orientation,
with compactified time $T$ pointing upward and radial null directions
at $45^\circ$ to the spatial coordinate $X$.
The same outgoing tail $h(v)$ is prescribed on $u=0$;
the colored segment on $v=0$ denotes the support of $f_n$, between
$u_*e^{-t_n-L_n}$ and $u_*e^{-t_n}$. Mass is compared on
$u=u_{\rm obs}>0$, along the interior null ray to the Cauchy horizon.
The open right corner $(0,\infty)$ is excluded from the local
extension statement. The dotted edge $u=U$ bounds the common region.
The initial segment
is schematic and is not drawn to scale. (b) Analytic envelope
densities for $t_n=L_n=nL_*$ and $n=1,2,4$, with a fixed large $L_*$.
Each curve has unit area, while
$\|f_n\|_p/\|f_1\|_p=n^{1/p-1/2}$.
The plotted density is $L_*B_n^2/E$;
the physical primitive energy approaches $E$ in the sequence limit.
Both figures use $\psi(y)=\exp(4-1/[y(1-y)])$ on $(0,1)$,
extended by zero.}
\end{figure*}

\textit{Model and critical regime.}
We use a neutral real massless scalar field, constant nonzero electric
charge $e$, and $\Lambda>0$. In spherical symmetry write
$g=-\Omega^2\dd u\dd v+r^2\dd\omega^2$.
The renormalized Hawking mass $\varpi$ is defined by
\begin{equation}\label{eq:mass-def}
 1-\frac{2\varpi}{r}+\frac{e^2}{r^2}-\frac{\Lambda r^2}{3}
 =-\frac{4r_ur_v}{\Omega^2}.
\end{equation}
Fix a nonextremal Reissner--Nordstr\"om--de Sitter reference with
surface gravities $K_+,K_->0$. The critical relation is
\begin{equation}\label{eq:critical-letter}
\begin{gathered}
a=2K_+,\qquad k=K_-=ab,\\
b=1-\frac1p,\qquad 4\le p<\infty.
\end{gathered}
\end{equation}
The equality $K_-/K_+=2(1-1/p)$ marks the boundary of the strict
bounded-mass criterion for rough transverse data with sufficiently
fast horizon decay; suitable rough data above it produce mass inflation
\cite[Theorem~1.1]{RossettiRough}.
The smooth packets below probe finite response at this boundary.
The outer-horizon rate $a$ converts logarithmic distance from the
initial corner into characteristic time. Integrating a transverse
profile of size $u^{-1/p}$ gives the power $u^b$, hence a decay rate
$ab$ along a fixed-radius matching curve. At \eqref{eq:critical-letter}
this rate equals the inner-horizon amplification rate $k$.
The remaining envelope then determines how much energy reaches the
horizon. This balance is sensitive to the exact exponents.
The transverse initial derivative is $f(u)=\phi_u(u,0)$.
The fixed gauge is $r_u(u,0)=-1$ and
$-\Omega^2/(4r_u)=1$ on the outgoing initial hypersurface.
In this gauge its scalar derivative lies in $L^p(\dd u)$, and the
constraint supplies a metric derivative in $L^{p/2}(\dd u)$.

On the outgoing side prescribe the same continuous positive tail
$h(v)=r\phi_v(0,v)$ for every solution, with
\begin{equation}\label{eq:h-letter}
 c_He^{-lv}\le h(v)\le C_He^{-sv},\qquad k<s<l<2s+k .
\end{equation}
The corner is chosen sufficiently late for the compatible horizon
construction and is then fixed. The zero-incoming reference has
$f=0$ and this same nonzero tail; it is a radiating solution.
The comparisons below use the same initial coordinates.

\textit{Concentrating inputs.}
Choose a smooth nonnegative shape $\psi$ supported on $[0,1]$,
positive inside, and set $M_2=\int_0^1\psi^2$.
Let $E>0$, and choose $t_n,L_n\to\infty$ with no restriction on
their relative growth. With $\tau(u)=\log(u_*/u)$, set
\begin{equation}\label{eq:packet-letter}
 f_n(u)=b\sqrt{\frac{E}{M_2L_n}}\,u^{-1/p}
                \psi\left(\frac{\tau(u)-t_n}{L_n}\right).
\end{equation}
The two geometric initial traces are reconstructed from the
constraints, as specified in the End Matter.
Figure~\ref{fig:spacetime}(a) locates these data and the observation
null ray in the local characteristic region.

The smallness of the data and the energy available for transmission
can be seen directly. For large~$n$,
\begin{equation}\label{eq:norm-letter}
 \|f_n\|_p^p=
 b^p(E/M_2)^{p/2}L_n^{1-p/2}\int_0^1\psi^p\longrightarrow0 .
\end{equation}
Transport in $u$ integrates the transverse derivative
before the blueshift amplifies it. The resulting critical energy is
\begin{equation}\label{eq:energy-letter}
 \Ep(f)=\int_0^{U_0}
 \left[u^{-b}\int_0^u f(x)\dd x\right]^2\frac{\dd u}{u}.
\end{equation}
Write the logarithmic envelope as
$B_n(t)=\sqrt{E/(M_2L_n)}\psi((t-t_n)/L_n)$.
The bracket in \eqref{eq:energy-letter} is its future exponential
average $b\int_0^\infty e^{-bz}B_n(\tau+z)\dd z$.
For broad packets this average approaches $B_n$ in $L^2$, giving
$\Ep(f_n)\to\|B_n\|_2^2=E>0$.
Figure~\ref{fig:spacetime}(b) displays how the envelope retains its
squared integral as it broadens. The initial $L^p$ norm decreases,
while the weighted energy selected by propagation remains finite.

\textit{Finite response theorem.}
\begin{theorem}\label{thm:letter}
For the data \eqref{eq:critical-letter}--\eqref{eq:packet-letter},
there is a common small $U>0$ on which the nonlinear developments
have uniformly bounded mass.
The initial $L^p$ perturbation and compatible initial metric
differences tend to zero. For every compact interval
$I\subset(0,U]$ there is a continuous $\mathcal A_h(u)>0$ such that,
uniformly in $u\in I$ and all sufficiently late observation times,
\begin{equation}\label{eq:response-letter}
 \Delta\varpi_n(u,v)=
 \mathcal A_h(u)F_\psi\left(\frac{av-t_n}{L_n}\right)+o(1),
\end{equation}
where $\Delta\varpi_n=\varpi_n-\varpi_0$ and
$F_\psi(y)=M_2^{-1}\int_{-\infty}^y\psi(z)^2\dd z$.
Each development admits a local weak extension at every
Cauchy-horizon point $0<u<U$ with continuous
$W^{1,2}_{\rm loc}$ metric, $L^2_{\rm loc}$ connection, and
$H^1_{\rm loc}$ scalar.
The horizon is locally the future Cauchy horizon of the original
two characteristic initial hypersurfaces.
\end{theorem}

In particular, at every fixed positive $u$,
\begin{equation}\label{eq:limits-letter}
 \begin{split}
 \lim_{n\to\infty}\lim_{v\to\infty}\Delta\varpi_n(u,v)
     &=\mathcal A_h(u)>0,\\
 \lim_{v\to\infty}\lim_{n\to\infty}\Delta\varpi_n(u,v)&=0 .
 \end{split}
\end{equation}
The second line follows from ordinary finite-time continuous
dependence. The first detects radiation that has escaped every fixed
time window. Thus the terminal mass map is discontinuous at the
radiating reference in the stated initial-data topology.

\begin{figure*}[t]
\centering
\includegraphics[width=170mm]{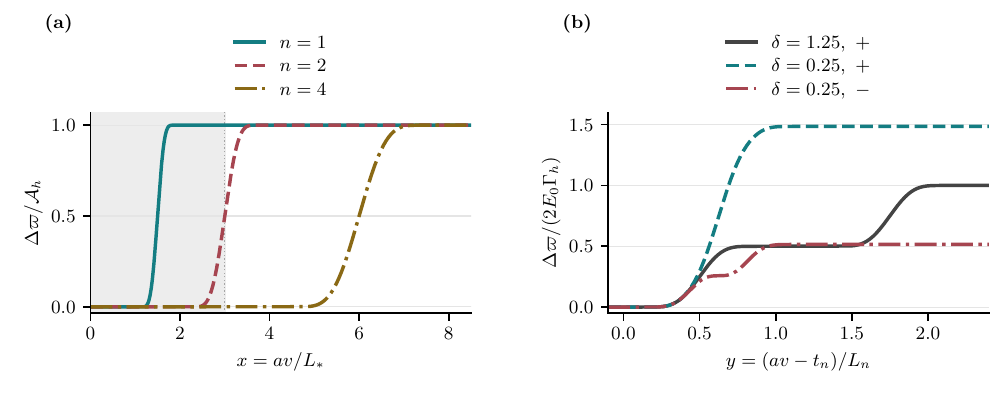}
\caption{\label{fig:response}Leading cumulative mass profiles.
(a) The sequence in Fig.~\ref{fig:spacetime}(b) gives
$\Delta\varpi/\mathcal A_h=F_\psi((x-n)/n)+o(1)$,
where $x=av/L_*$. The shaded interval $0\le x\le3$ is a fixed
observation window. The response moves out of it as $n$ grows,
while the terminal amplitude remains unchanged.
(b) Two packets of equal component energy $E_0$, common width
$L_n$, separation $\delta L_n$, and relative sign $\sigma$ give
$\Delta\varpi/(2E_0\Gamma_h)=G_{\delta,\sigma}(y)+o(1)$,
where $y=(av-t_n)/L_n$ and
$G_{\delta,\sigma}(y)=(2M_2)^{-1}\int_{-\infty}^y
[\psi(z)+\sigma\psi(z-\delta)]^2\dd z$.
All three curves share the same normalization. Separated packets give
two increments with an intervening plateau; overlap enhances or
suppresses the terminal response according to the relative sign.}
\end{figure*}

\textit{From radiation to mass.}
Propagation converts the retained energy into mass through a radiation
gain and a geometric factor. Write $\theta=r\phi_v$.
An exponentially conjugated radial mode has a finite positive inner
amplitude $\mathcal T$. It determines the gain for a slowly varying
envelope. As the envelope broadens, its variation over a fixed
propagation interval decreases. The radial mode therefore transmits
its leading shape with a constant gain. Establishing this comparison
on the evolving geometry requires an estimate that follows the
packet as its arrival time increases. The actual nonlinear radiation
satisfies
\begin{equation}\label{eq:radiation-letter}
 \left\|e^{kv}(\theta_n-\theta_0)
              -C_h\mathcal T B_n(av)\right\|_{L^2_vC^0_u(I)}
 \longrightarrow0 .
\end{equation}
Here $C_h>0$ is fixed by the outgoing tail and the initial gauge.
The comparison field is this radial mode multiplied by $B_n(av)$.
Its transport error is proportional to $\|B_n'\|_2=O(L_n^{-1})$.
A finite phase has a vanishing effect on this broad-packet limit.

The norm in \eqref{eq:radiation-letter} controls the full time axis.
It therefore determines the total transmitted energy as well as
the part seen during the main response. The End Matter outlines the estimate.

The exact flux identities, with $\nu=r_u$ and $D=-r_v/\Omega^2$, are
\begin{equation}\label{eq:massflux-letter}
 D_v=\frac{\theta^2}{r\Omega^2},\qquad
 \varpi_v=(-2r\nu)D_v .
\end{equation}
Let $R_h$, $N_h$, and $A_h$ be the zero-incoming limits of
$r$, $\nu$, and $e^{2kv}\Omega^2$, respectively. These describe the
horizon radius, its transverse variation, and the coefficient of conformal
decay. They satisfy
$R_h,A_h>0$ and $N_h<0$ on~$I$. The actual corresponding geometry
converges to these reference quantities even though
$\Ep(f_n)$ does not vanish.
The transmitted radiation contributes quadratically to the mass flux.
Integrating the exact identities \eqref{eq:massflux-letter} gives
\begin{equation}\label{eq:amplitude-letter}
 \mathcal A_h(u)=
 \frac{-2N_h(u)}{aA_h(u)}\,C_h^2\mathcal T^2E .
\end{equation}
The positive outgoing tail ensures $N_h<0$ through the transverse
constraint. Thus the radiating reference supplies a strictly positive
geometric conversion factor.

The coexistence of geometric convergence and a surviving mass response
is visible in \eqref{eq:massflux-letter}. The normalized null expansion
$D$ accumulates the radiation flux. Its contribution to $r_v=-\Omega^2D$
is suppressed by the decaying conformal factor, whereas its contribution
to the mass contains the finite factor $-2r\nu$.
Consequently, convergence of the radius and the conformal decay
coefficient can accompany a finite change in terminal mass.

\textit{Delay and duration.}
The support of \eqref{eq:packet-letter} lies between
$u_*e^{-t_n-L_n}$ and $u_*e^{-t_n}$.
Its distance from the initial corner is exponentially small in
$t_n$, while its arrival time is linear in $t_n/a$.
Equation \eqref{eq:response-letter} supplies the entire accumulated
response, so the time at which a fixed fraction $q$ is first reached
can be read from $F_\psi$.
For $q$ in a compact subinterval of $(0,1)$,
\begin{equation}\label{eq:quantiles-letter}
 v_{q,n}=\frac{t_n+L_nF_\psi^{-1}(q)}{a}+o_I(L_n).
\end{equation}
For a symmetric shape, the median is
$(t_n+L_n/2)/a+o(L_n)$ and the central duration is
$L_n[F_\psi^{-1}(0.9)-F_\psi^{-1}(0.1)]/a+o(L_n)$.
Position and width are independent controls;
Fig.~\ref{fig:response}(a) illustrates a sequence in which both grow.
The times refer to the fixed characteristic gauge;
the mass difference itself is the scalar observable
\eqref{eq:mass-def} evaluated on corresponding rays.

\textit{Two packets and their relative sign.}
The cumulative law also resolves how successive inputs combine.
Take two smooth envelopes $B_{1,n},B_{2,n}$ of the preceding form,
with fixed component energies $E_1,E_2$ and with each logarithmic
position and width tending to infinity. Prescribe the total
transverse input through $B_{1,n}+\sigma B_{2,n}$, where
$\sigma=\pm1$, and reconstruct its compatible geometry from the
constraints. The same energy estimate applies to the combined envelope
and gives, uniformly on $I$ and over all late cumulative endpoints,
\begin{equation}\label{eq:pair-letter}
 \begin{split}
 \Delta\varpi_n(u,v)&=\Gamma_h(u)
  \int_{-\infty}^{av}(B_{1,n}+\sigma B_{2,n})^2\dd t+o_I(1),\\
 \Gamma_h(u)&=\mathcal A_h(u)/E>0 .
 \end{split}
\end{equation}
The coefficient depends on the common radiating reference.
The two scalar components propagate on the same evolving metric,
and their product enters the gravitational source.

Separated envelopes give two successive increments with an
intermediate plateau in the leading profile. When the envelopes
overlap, their relative sign controls the contribution
$2\sigma\int B_{1,n}B_{2,n}$.
Figure~\ref{fig:response}(b) compares these possibilities at equal
component energies. Same-sign overlap increases the terminal
increment, whereas opposite signs reduce it. Holding the two
component energies fixed in this comparison allows the energy of
the combined input to vary with their overlap. Identical envelopes
with opposite signs give exactly the reference solution.

\textit{Discussion.}
Our nonlinear family combines bounded mass and weak extendibility
with discontinuous dependence of the terminal mass on the initial
$L^p$ data. The surviving weighted energy records information lost by
this norm at the critical relation and provides a natural quantity
for comparing terminal responses.

Finite observation windows can conceal late behavior in nonlinear
black-hole interiors \cite{Luna2019,Luna2021}.
Finite amplification also occurs near evolving inner trapping horizons
\cite{CarballoRubio2024}.
Our cumulative law quantifies when the response becomes visible:
observation times must follow $(t_n+yL_n)/a$ for $y$ inside the envelope.
In a regular coordinate with the horizon at $V=0$, these late intervals
accumulate at that boundary. Initial concentration thus becomes
accumulation near the future horizon.

The characteristic problem isolates critical concentration.
Connecting these inputs to exterior collapse would determine which
packets arise dynamically. Rapidly varying envelopes would probe the
frequency dependence of transmission and the reach of the response
law beyond slowly varying inputs.

\begin{acknowledgments}
\textit{Acknowledgments.}
GPT-6 Astra was used for the Lean 4 implementation and language editing.
The author checked the initial assumptions and final theorem statements
against the manuscript and confirmed their agreement. Some proof
strategies suggested by the model to facilitate formalization were
incorporated into the manuscript to align the written and formalized
proofs.
\end{acknowledgments}

\textit{Data availability.}
The complete Lean 4 formalization of the mathematical results and proofs
in this work, together with verification instructions, is available at
\url{https://github.com/l46430640-del/Finite-Mass-Response-Formal}
(version v0.1.0).

\bibliographystyle{apsrev4-2}
\bibliography{references}
\onecolumngrid
\clearpage
\twocolumngrid
\section*{End Matter}
\input{endmatter}

\onecolumngrid
\clearpage
\setcounter{equation}{0}
\setcounter{section}{0}
\setcounter{subsection}{0}
\makeatletter
\let\arxivoddhead\@oddhead
\let\arxivevenhead\@evenhead
\def\@oddhead{{\fontsize{10}{12}\selectfont\arxivoddhead}}
\def\@evenhead{{\fontsize{10}{12}\selectfont\arxivevenhead}}
\renewcommand\normalsize{\@setfontsize\normalsize{10.95pt}{13.6pt}}
\makeatother
\normalsize
\renewcommand{\theHequation}{supp.\arabic{equation}}
\renewcommand{\theHsection}{supp.\arabic{section}}
\renewcommand{\theHsubsection}{\theHsection.\arabic{subsection}}
\section*{Supplemental Material}

\renewcommand{\theequation}{S\arabic{equation}}
\renewcommand{\thesection}{S\arabic{section}}
\renewcommand{\thesubsection}{\thesection.\arabic{subsection}}
\setcounter{secnumdepth}{2}
\newcommand{\cA}{\mathscr A}

We prove the finite mass response, its cumulative profile, and local
weak extendibility. The proof first constructs a common nonlinear
interior, then compares the radiation with a single positive radial
mode, and finally integrates the exact mass flux. The estimates include
a finite sum of signed broad packets.

\input{sections/setup}
\input{sections/common}
\input{sections/radiation}
\input{sections/response}
\input{sections/local}

\end{document}

%% file: macros.tex
\newcommand{\R}{\mathbb R}

\newcommand{\Ep}{\mathcal E_p}

\newcommand{\dd}{\,\mathrm d}

\newcommand{\one}{\mathbf 1}

%% file: endmatter.tex
\textit{Compatible data.}
Our scalar convention is
$r_{uu}-(\log\Omega^2)_ur_u=-r\phi_u^2$, and similarly in $v$.
Let $r_\bullet=r_H(0)$. The incoming data are
\begin{align}
 r(u,0)&=r_\bullet-u,\nonumber\\
 \phi(u,0)&=\phi_H(0)+\int_0^u f_n(x)\dd x,\label{eq:initial-end}\\
 \Omega^2(u,0)&=4\exp\left[-\int_0^u(r_\bullet-x)f_n(x)^2\dd x\right].
 \nonumber
\end{align}
On $u=0$ set $\varpi_H=M-\frac12\int_v^\infty h^2$ and solve
$r_H'=1-2\varpi_H/r_H+e^2/r_H^2-\Lambda r_H^2/3$ with
$r_H(\infty)=r_+$. Then
$\nu_H=-\exp\int_0^v2K_H$, $\Omega_H^2=-4\nu_H$, and
$\phi_H=-\int_v^\infty h/r_H$, where
$K_H=\varpi_H/r_H^2-e^2/r_H^3-\Lambda r_H/3$.
At a sufficiently late corner the terminal radius equation is a
contraction near $r_+$.
The reference uses the same data with $f_n=0$.

Writing $d_n=\|f_n\|_p$, the incoming scalar and conformal logarithm
differences are bounded by $U_0^bd_n$ and
$r_\bullet U_0^{1-2/p}d_n^2$. The latter derivative has
$L^{p/2}$ norm at most $r_\bullet d_n^2$. Thus $p\ge4$ places the
data in the characteristic finite-energy local theory
\cite{GajicLuk2019}. H\"older's bounds
$\int_0^u|f_n|\le M_pu^b$ and
$\int_0^u f_n^2\le M_p^2u^{1-2/p}$, where $M_p=\sup_n d_n$,
give uniform smallness at the corner.

\textit{The common interior.}
The retained weighted energy controls the scalar field and the
geometry together on a common region of evolution.
Near the outer horizon the
scalar transport has the strict margin $a-k>0$. Across the interior,
integration over thin radius layers bounds the coupling between the
two null derivatives. Near the inner horizon this coupling becomes
small as the matching radius approaches $r_-$.

The critical relation requires the exact conformal decay rate.
The radial equation near $r_-$ has a positive restoring coefficient;
integrating it controls the accumulated error in that rate.
Matching to a fixed-radius curve $u\asymp e^{-av}$ then gives
$\Omega^2\asymp u^{-1-2b}e^{-2kv}$ in the late region.
The resulting radiation estimate holds on every finite time interval
with a constant depending on the total energy budget.
The nonlinear feedback is multiplied by a positive power of $u$,
so choosing one sufficiently small $U$ closes the estimate for the
whole sequence. The mass identity gives the uniform mass bound.
This argument accommodates radiation whose energy moves to
arbitrarily late times.

\textit{A positive mode.}
Put $F_0=1-2M/r+e^2/r^2-\Lambda r^2/3$.
Define $z'/z=a/F_0$ with $z/(r_+-r)\to1$, and set $q_0=z^b$.
In the inward coordinate $x=r_+-r$, solve
\begin{equation}
 X_x=\frac{a q_0V}{r},\qquad V_x=\frac{X}{arq_0},\qquad
 (X,V)|_{x=0}=(0,r_+b).
\end{equation}
Both coefficients are integrable: $q_0$ behaves as
$(r_+-r)^b$ at the outer endpoint and $(r-r_-)^{-1/2}$ at the
inner one. The positive integral equations give
$\mathcal T=X(r_-)\in(0,\infty)$.
For $W=q_0V$ they also give
$X(Y)=\mathcal T+O(\sqrt{Y-r_-})$ and
$(-F_0(Y))W(Y)=O(\sqrt{Y-r_-})$.
The horizon data fix
\begin{equation}
 c_\nu=\exp\int_0^\infty2(K_H-K_+)\dd v,\qquad
 C_h=\frac{r_\bullet}{r_+c_\nu^b}.
\end{equation}

\textit{Radiation comparison.}
Write $\lambda=r_v$ and $\zeta=r\phi_u$.
On a fixed radius band normalize the total scalar fields as
$e^{kv}\theta$ and $e^{kv}u\zeta/q_0$ on each actual metric.
The reference coupling coefficients are $aq_0/r$ and $1/(arq_0)$,
both integrable on the full radius interval. Gronwall's inequality
therefore gives a time $L^2$ bound uniform in both radius cutoffs.
The actual coefficient errors, including the time derivative of the
characteristic speed, tend to zero at late times.
Comparison with $(X,V)C_h B_n(av+d_R)$, where
$d_R=\log(c_\nu u_*)-\log z(R)$, leaves those errors and $B_n'$.
A fixed phase change also costs only a constant times $\|B_n'\|_2$.

Apply this energy estimate on $[T,\infty)$, keeping its inflow at
$v=T$. At fixed $T$ the inflow converges to the reference inflow
by finite-time continuous dependence. The reference band traces
decay exponentially, since $\min(a,s)>k$. First take the packet
limit, then $T\to\infty$. The outer matching error can be made
arbitrarily small as its radius approaches $r_+$;
the inner error vanishes by $(-F_0(Y))W(Y)\to0$.
Late transverse integration has decaying kernels, which control the
supremum over $I$ inside the time integral.
Subtract the reference radiation and join this late comparison
to finite-time convergence on the fixed history. This proves
\eqref{eq:radiation-letter} on the entire observation half-line.
See Sec.~S3 of the Supplemental Material for the proof.

\textit{Accumulation.}
The identity
$(r\nu)_v=r\Omega^2[-1/(4r)+e^2/(4r^3)+\Lambda r/4]$
gives common tails for $\nu$, in addition to those for $r$ and
$e^{2kv}\Omega^2$. Finite-time convergence thus determines their
boundary limits without requiring radiation-energy convergence.
The positive outgoing tail generates $\zeta_{0,\infty}>0$;
the transverse constraint then gives $N_h<0$.

In $\theta_n^2-\theta_0^2$, the cross term is controlled by local
convergence on $[v_0,T]$ and the reference radiation tail on
$[T,\infty)$. Strong $L^2$ comparison controls the remaining square
in $L^1$. Substitution in \eqref{eq:massflux-letter} and
$\dd t=a\dd v$ gives \eqref{eq:amplitude-letter} and controls
every cumulative upper endpoint. Inverting the strictly increasing
part of $F_\psi$ gives \eqref{eq:quantiles-letter}.
For two signed packets, use their sum in the scalar equation and
its full square in the constraints; the same comparison proves
\eqref{eq:pair-letter}.

\textit{Weak extension.}
The coordinate $\dd V/\dd v=\Omega^2(U,v)$, with $V(\infty)=0$,
is regular at the horizon. Finite radiation energy gives square-integrable
$V$ derivatives and continuous transverse traces. Compatible
characteristic continuation across $V=0$ yields the weak extensions
stated in the theorem.

%% file: sections/setup.tex
\section{Equations, data, and packet norms}\label{sec:setup}
Fix $e\ne0$, $\Lambda>0$, and a mass $M$ for which
$F_0(r)=1-2M/r+e^2/r^2-\Lambda r^2/3$ has simple positive inner,
event, and cosmological roots $r_-<r_+<r_c$.
Put $a=F_0'(r_+)=2K_+$, $k=-F_0'(r_-)/2=K_-$, and assume
\begin{equation}\label{eq:critical}
 b=1-1/p=k/a,\qquad 4\le p<\infty,\qquad \rho=2b.
\end{equation}
All constants below may depend on this fixed background. The time
normalization is fixed by $\kappa=1$ on $u=0$.

For $g=-\Omega^2\dd u\dd v+r^2\dd\omega^2$, define
\begin{align}\label{eq:variables}
 &\nu=r_u,\quad\lambda=r_v,\quad\theta=r\phi_v,\quad\zeta=r\phi_u,
 \quad\ell=\log\Omega^2,\nonumber\\
 &F=F_0-2(\varpi-M)/r,\quad
 K=\varpi/r^2-e^2/r^3-\Lambda r/3,\quad
 \kappa=-\Omega^2/(4\nu).
\end{align}
The equations in the region $\nu<0$ are
\begin{align}
 &r_u=\nu,\quad r_v=\lambda=\kappa F,\quad
 \nu_v=\lambda_u=2\kappa K\nu,\quad
 \kappa_u=\kappa\zeta^2/(r\nu),\label{eq:first}\\
 &\theta_u=-\lambda\zeta/r,\qquad
 \zeta_v=-\nu\theta/r,\qquad
 \varpi_v=\theta^2/(2\kappa),\qquad
 \varpi_u=F\zeta^2/(2\nu).\label{eq:scalar}
\end{align}
Equivalently, with
$P(r)=-1/(4r)+e^2/(4r^3)+\Lambda r/4$ and
$Q(r)=1/(2r^2)-e^2/r^4$,
\begin{align}
 r_{uv}&=\Omega^2P-\nu\lambda/r,&
 \phi_{uv}&=-(\nu\phi_v+\lambda\phi_u)/r,\nonumber\\
 \ell_{uv}&=-2\phi_u\phi_v+\Omega^2Q+2\nu\lambda/r^2,&
 r_{uu}-\ell_ur_u&=-r\phi_u^2,
 \quad r_{vv}-\ell_vr_v=-r\phi_v^2.\label{eq:wave}
\end{align}
In the trapped region put
\begin{equation}\label{eq:flux}
 D=-\lambda/\Omega^2>0,\qquad q=(4D)^{-1}=\nu/F,
 \qquad D_v=\frac{\theta^2}{r\Omega^2},\qquad
 \varpi_v=(-2r\nu)D_v,\qquad (r\nu)_v=r\Omega^2P.
\end{equation}

\subsection{Compatible characteristic data}
Fix a continuous tail
\begin{equation}\label{eq:tail}
 0<c_He^{-lv}\le h(v)\le C_He^{-sv},\qquad k<s<l<2s+k.
\end{equation}
Choose a sufficiently late common corner and translate it to $v=0$.
The constants in \eqref{eq:tail} refer to this translated gauge.
On $u=0$ prescribe
\begin{align}\label{eq:horizon}
 &\varpi_H=M-\tfrac12\int_v^\infty h^2\dd w,\quad
 r_H'=F(r_H,\varpi_H),\quad r_H(\infty)=r_+,\nonumber\\
 &\nu_H=-\exp\int_0^v2K_H\dd w,\quad
 \Omega_H^2=-4\nu_H,\quad
 \phi_H=-\int_v^\infty h/r_H\dd w.
\end{align}
For $z_H=M-\varpi_H$ and $y_H=r_+-r_H$, put
$\mathcal N(y)=ay+F_0(r_+-y)$. The terminal equation is
\begin{equation}\label{eq:terminal}
 y_H(v)=\int_v^\infty e^{-a(w-v)}
 \left\{\mathcal N(y_H(w))+
 \frac{2z_H(w)}{r_+-y_H(w)}\right\}\dd w.
\end{equation}
The rescaling $y_H=e^{-2sv}Y$ gives a future kernel of rate $a+2s$.
Since $\mathcal N(0)=\mathcal N'(0)=0$, its integral map is a
contraction with constant at most $1/2$ on a small closed ball in
$\|Y\|_\infty$ once the corner is sufficiently late.
Thus $y_H$, $\varpi_H-M$, and $K_H-K_+$ are
$O(e^{-2sv})$. Differentiating the radial equation gives
\begin{equation}\label{eq:lambda-h}
 \lambda_H(v)=\int_v^\infty
 e^{-\int_v^w2K_H\dd z}\frac{h(w)^2}{r_H(w)}\dd w>0.
\end{equation}
No derivative of $h$ is required. In particular
$c_\nu=\exp\int_0^\infty2(K_H-K_+)\dd v$ is finite and positive.

With $r_\bullet=r_H(0)$ and $0<U_0<r_\bullet$, the other side is
\begin{equation}\label{eq:data}
 r(u,0)=r_\bullet-u,\qquad
 \phi(u,0)=\phi_H(0)+\int_0^u f(x)\dd x,\qquad
 \ell(u,0)=\log4-\int_0^u(r_\bullet-x)f(x)^2\dd x.
\end{equation}
These formulas satisfy both corner compatibility and the incoming
constraint. The zero-incoming reference, denoted by a subscript $0$,
uses $f=0$ and the same $h$.

\subsection{The packet family}
Let $\psi\in C_c^\infty(\R)$ have support $[0,1]$ and be positive
on $(0,1)$, and set $M_2=\int\psi^2$. Fix $u_*>0$ and $E>0$, and define
\begin{equation}\label{eq:packets}
 B_n(t)=\sqrt{\frac{E}{M_2L_n}}\psi\left(\frac{t-t_n}{L_n}\right),
 \quad f_n(u)=bu^{b-1}B_n(\tau(u)),\quad
 \tau(u)=\log(u_*/u),\quad t_n,L_n\longrightarrow\infty.
\end{equation}
No relation between the two divergent sequences is assumed. A finite
prefix may be retained after restriction to $[0,U_0]$.
The change of variables $u=u_*e^{-t}$ gives, for all sufficiently
large $n$,
\begin{equation}\label{eq:norms}
 \|f_n\|_p^p=b^p(E/M_2)^{p/2}L_n^{1-p/2}\int\psi^p\to0,
 \quad\|B_n\|_2^2=E,\quad
 \|B_n'\|_2=\frac{\sqrt{E/M_2}}{L_n}\|\psi'\|_2\to0.
\end{equation}
The same properties, with a bounded $L^2$ norm, hold for any fixed
finite sum of these packets with either sign. In the proofs $B_n$
may denote such a sum. It has a common $L^\infty$ bound $A_*$,
its left support tends to infinity, and $f_n\to0$ in $L^p$.
Writing $M_p=\sup_n\|f_n\|_p<\infty$, H\"older gives
\begin{equation}\label{eq:data-bounds}
 \int_0^u|f_n|\le M_pu^b,\qquad
 \int_0^u f_n^2\le M_p^2u^{1-2/p}.
\end{equation}
Both powers are positive and give a common rate of smallness at
the corner.

The primitive and its energy are
$G_n(u)=u^{-b}\int_0^u f_n$ and
$\Ep(f_n;U_0)=\int_0^{U_0}G_n(u)^2\dd u/u$.
In logarithmic coordinates,
\begin{equation}\label{eq:average}
 G_n(u)=u^{-b}\int_0^u f_n(x)\dd x
       =(H_bB_n)(\tau(u)),\qquad
 H_bB(t)=b\int_0^\infty e^{-bz}B(t+z)\dd z.
\end{equation}
Minkowski's inequality gives $\|H_b B\|_2\le\|B\|_2$ and
$\|H_bB-B\|_2\le b^{-1}\|B'\|_2$.
For large $n$, $B_n$ vanishes below $\tau_* =\log(u_*/U_0)$.
Restriction to $[\tau_*,\infty)$ is an $L^2$ contraction, so
\begin{equation}\label{eq:restricted-energy}
 \left|\Ep(f_n;U_0)-\|B_n\|_2^2\right|
 \le \frac{2}{b}\|B_n\|_2\|B_n'\|_2\longrightarrow0.
\end{equation}
Thus for a single packet $\Ep(f_n;U_0)\to E$; for a signed sum it
differs from $\|B_n\|_2^2$ by $o(1)$.
The constraints give initial $C^0$ differences bounded by
$U_0^b\|f_n\|_p$ and $r_\bullet U_0^{1-2/p}\|f_n\|_p^2$ for
$\phi$ and $\ell$, respectively, while
$\|\partial_u(\ell_n-\ell_0)\|_{p/2}\le r_\bullet\|f_n\|_p^2$.

%% file: sections/common.tex
\section{A common nonlinear interior}\label{sec:common}
The following construction uses only the bounds in
\eqref{eq:norms}--\eqref{eq:data-bounds}, the common pointwise coefficient
$|f_n(u)|\le bA_*u^{-1/p}$, and the fixed outgoing data.
Let $E_*\ge\sup_n\|B_n\|_2^2$.

The main issue is uniform control as the radiation moves to later
times. Coarse decay first locates the fixed-radius matching curves
and keeps the area radius positive. Integrating the restoring term
in the radial equation then recovers the exact conformal exponent
$2k$, needed to measure the critical radiation energy.
The scalar transport estimate uses this weight on finite time
intervals. Its feedback into the null expansion is small near the
corner, which gives one region of existence for the whole family.
Finally, integrable tails for the geometric variables connect
finite-time continuous dependence to their horizon limits.

\begin{proposition}\label{prop:common}
There is one $U>0$ such that all the characteristic developments,
including the zero-incoming reference, exist on
$[0,U]\times[0,\infty)$ with $r\ge r_*>0$ and uniformly bounded
$\varpi$. On every compact $I\subset(0,U]$, for some fixed $v_I$,
\begin{equation}\label{eq:energy-bound}
 \int_{v_I}^\infty e^{2kv}\sup_{u\in I}|\theta_n(u,v)|^2\dd v
 \le C_I\mathcal B^2,\qquad
 \mathcal B^2=E_*+M_p^2+\frac{C_H^2}{2(s-k)}.
\end{equation}
Writing $\cA_n=e^{2kv}\Omega_n^2$, the fields
$r_n,\cA_n,\nu_n$ have common convergence tails at $v=\infty$,
and converge uniformly to their zero-incoming counterparts on the
compactified strip $I\times[v_I,\infty]$.
In particular $r_n,\cA_n$ have common positive lower bounds there.
The choice of $U$ may depend on $E_*$ and the sequence through
$M_p,A_*$, but is independent of the observation time.
\end{proposition}

\begin{proof}
The local construction in Sec.~\ref{sec:local} supplies existence,
constraint propagation, and continuation on finite rectangles.
We now obtain estimates independent of the terminal time.

\textit{Coarse geometry.}
The division into an outer neighborhood, finite radius
bands, and two inner regions is the characteristic construction of
Refs.~\cite{Costa2018,Rossetti2025,RossettiRough}; at the critical
exponent we will additionally integrate the radial restoring term.
Choose $h_c=k-\epsilon_c>0$ with $l<2s+h_c$, and put
$h_e=(k+h_c)/2$. Choose $\beta>0$ sufficiently small that
$h_c+(k+h_c)\beta<h_e$,
and take the geometric losses smaller than $\rho\beta/2$.
We may weaken the resulting positive transverse power to
\begin{equation}\label{eq:losses}
 \alpha=1+\beta<9/8,\qquad
 \delta_\nu=\rho\beta/2=b\beta>0.
\end{equation}
All losses in this paragraph are fixed, rather than sent to zero.

In a sufficiently thin outer neighborhood, integration of
\eqref{eq:first} and \eqref{eq:scalar} gives
\begin{equation}\label{eq:outer-coarse}
 -\nu\asymp e^{av},\quad
 |\lambda|\le C(e^{-2sv}+ue^{av}),\quad
 |\theta|\le Ce^{-h_ev},\quad
 |\zeta|\le C(|f(u)|+e^{(a-h_e)v}).
\end{equation}
For example, inserting the vertical scalar integral into the
horizontal one bounds its free term by
$Ce^{av}\int_0^u x|f(x)|\dd x
 +Ce^{-2sv}\int_0^u|f|+h$, which is $O(e^{-kv})$ on
$u\lesssim e^{-av}$ by \eqref{eq:data-bounds}.
The twofold integral has an arbitrarily small coefficient as the
neighborhood is made thin; conjugation by $e^{h_ev}$ leaves the
strict margin $a-h_e>0$.
The integrated constraint is
\begin{equation}\label{eq:kappa-bound}
 -\log\kappa=\int_0^u\frac{\zeta^2}{r(-\nu)}\dd x
 \le C\{M_p^2e^{-av}u^{1-2/p}+u e^{(a-2h_e)v}\}.
\end{equation}
The mass equation gives the corresponding quadratic smallness.
Thus the geometry improves the bounds used in the scalar integrals.
This is a finite-time bootstrap starting from \eqref{eq:horizon}.

Let $R$ be the outer radius, fix $r_-<Y<R$, and write
$r(u_Y(v),v)=Y$. Outer smallness on one late slice and
$\varpi_v\ge0$ give $\varpi\ge M-\epsilon$ thereafter.
Choose $\epsilon$ so that $F(r,M-\epsilon)<0$ on $[Y,R]$.
Under the bootstrap $\kappa\ge1/2$, this gives
$-\lambda\ge\lambda_*>0$ and a bounded residence time in the band.
The equation $\partial_v\log(r(-\nu))=-4\kappa P$ then propagates
the outer comparison $-\nu\asymp e^{av}$. The horizontal radial
flux bounds $|\lambda|$, while the outer arrival scale and residence
time give $Ue^{aV}\le\delta_{R,Y}$ for characteristic rectangles
with terminal corner in the band.
These constants depend on $R,Y$, but not on $V$.

Set $F_V(u)=\sup_{0\le w\le V}e^{h_ew}|\phi_v(u,w)|$.
Since $h_e<a(1-1/p)$, H\"older's inequality bounds
$e^{h_eV}\int_0^U|f|$ uniformly on these rectangles.
The vertical and horizontal scalar integrals therefore give
\[
 |\phi_u(u,w)|\le C|f(u)|+CF_V(u)e^{(a-h_e)w},\qquad
 F_V(u)\le C+Ce^{aV}\int_0^u F_V(x)\dd x.
\]
Gronwall yields $F_V\le C\exp(C\delta_{R,Y})$, with no smallness
requirement on the band width. Substitution into
\eqref{eq:kappa-bound} gives
$-\log\kappa\le C_{R,Y}(e^{-av}+e^{-2h_ev})$.
Taking the entry time late improves $\kappa\ge1/2$ to
$\kappa\ge3/4$. Finite-rectangle continuation propagates
\eqref{eq:outer-coarse} to $Y$, and gives
\begin{equation}\label{eq:arrival}
 u_Y(v)\asymp e^{-av},\quad
 0<c_Y\le-\frac{u_Y'(v)}{u_Y(v)}\le C_Y,\quad
 (-\nu_Y)(-u_Y')=-\lambda_Y.
\end{equation}
The last equality follows by differentiating the radius level and
will be used as an exact Jacobian.

Let $v_Y$ be the inverse arrival function, $s_u=v_Y(u)$,
$\Omega_Y^2(u)=\Omega^2(u,s_u)$, and $v_\gamma(u)=\alpha s_u$.
In the early inner region $s_u\le v\le\alpha s_u$, the same
integral bootstrap gives $r$ in a fixed neighborhood of $r_-$,
$1/2\le\kappa\le1$, $|\theta|\le Ce^{-h_ev}$, and,
for some $h_g,a_e>0$,
\begin{equation}\label{eq:early-coarse}
 |\varpi-M|+|\ell_v-2K|\le C_Ye^{-h_gv},\qquad
 J_Y(u,v)\le C_Y\beta s_u e^{-a_e s_u}=:b_\gamma(u),
 \quad J_Y=\int_{s_u}^v\frac{\theta^2}{r(-\lambda)}\dd w.
\end{equation}
The radius bound follows directly from the mass flux. Fix
$0<r_{\mathrm b}<r_-$ close to $r_-$, so $F_0(r_{\mathrm b})>0$.
The scalar bound, $\kappa\ge1/2$, and the matching mass estimate
make $|\varpi-M|$ uniformly small on $[s_u,\alpha s_u]$ for late
$s_u$. Once $\varpi-M\le r_{\mathrm b}F_0(r_{\mathrm b})/4$,
a crossing of $r=r_{\mathrm b}$ would give
$F(r_{\mathrm b},\varpi)\ge F_0(r_{\mathrm b})/2>0$.
This contradicts $\lambda=\kappa F<0$, which persists after entry
by the outgoing constraint. Hence $r_{\mathrm b}<r\le Y$.
Choosing $r_{\mathrm b}$ and $Y$ near $r_-$ and then taking the
entry time late improves $|K+k|$ by continuity at $(r_-,M)$.

To see the strict margins here, divide the square bound
$|\theta|^2\le Ce^{-2h_ev}$ by the coarse inner estimate
$-\lambda\gtrsim e^{-(2k+\epsilon)(v-s_u)}$.
On $v-s_u\le\beta s_u$ its integral is bounded by
$C_Y\beta s_u\exp[-(2h_e-(2k+\epsilon-2h_e)\beta)s_u]$.
The exponent is positive for the stated order of choices.
The constraint integral also has a positive decay margin.
For the conformal estimate it is useful to integrate the exact
identity, on any horizontal segment,
\begin{equation}\label{eq:ell-identity}
 [\ell_v-2K]_{u_1}^{u_2}
 =[2K(\kappa-1)]_{u_1}^{u_2}
 -\int_{u_1}^{u_2}
 \left(2K\kappa_u+\frac{2\kappa\varpi_u}{r^2}
                         +2\phi_u\phi_v\right)\dd u.
\end{equation}
Its sources are controlled by \eqref{eq:data-bounds}; it preserves the
partial derivative $\ell_v$ before evaluation on $\gamma$.

In the late region $v\ge v_\gamma(u)$, the coarse estimates are
\begin{align}\label{eq:late-coarse}
 &r\ge r_*>0,\qquad -\lambda\le Ce^{-\eta_\lambda v},\qquad
 -\nu\le Cu^{-1+\delta_\nu},\qquad |\phi_v|\le Ce^{-h_cv},
 \nonumber\\
 &|\phi_u(u,v)|\le C|f(u)|+Cu^{-A},
 \qquad A=1-\delta_\nu-h_c\alpha/a<1.
\end{align}
The transverse estimate follows directly from $\zeta_v=-\nu\theta/r$.
Before $\gamma$, the outer and band estimates, followed by the
coarse early bound $-\nu\le C u^{-1}e^{-c(v-s_u)}$ with $c>0$,
bound the vertical history by $Cu^{-1+h_e/a}$.
After $\gamma$ it is bounded by
$Cu^{-1+\delta_\nu+h_c\alpha/a}$.
Our choice gives $\delta_\nu+h_c\alpha/a<h_e/a<b$, so the latter
power bounds both histories. Since $r$ stays in a fixed positive
range, the direct incoming term costs only $C|f(u)|$.
Bootstrap $|\phi_v|\le H_Le^{-h_cv}$ and
$\ell_v\le-(k+h_c)$. On the null segment $v=v_\gamma(z)$, the outgoing
constraint gives
\[
 |\lambda(u,v)|\le C(z/u)^\sigma+CH_L^2e^{-2h_cv},
 \qquad \sigma=(k+h_c)\alpha/a>b,\quad \sigma+A>1.
\]
The two horizontal integrals are bounded by H\"older and
$\int u^{-\sigma-A}\dd u$, respectively. The upper arrival
scale converts $z^{h_c\alpha/a}$ into $Ce^{-h_cv}$, while
$z^b$ gains the additional factor $U^{b-h_c\alpha/a}$.
The remaining feedback has a factor $CU^{\delta_\nu}$, and
the inhomogeneous term in the $\lambda$ bound has the positive width
powers $U^b$ and $U^{1-A}$. Fixing $H_L$ before shrinking $U$
therefore improves the scalar bound by a factor of two.
Integration of $-\nu$ costs $CU^{\delta_\nu}/\delta_\nu$;
shrinking $U$ keeps $r\ge r_-/2$ and closes the radius bootstrap.
For the conformal bound, $\kappa\le1$ gives $\Omega^2\le4(-\nu)$.
The conformal wave equation, with the fixed scalar and $\lambda$
bounds, then gives
\[
 |\ell_v(u,v)-\ell_v(z,v)|
 \le C\bigl(U^{\delta_\nu}+M_pU^b+U^{1-A}\bigr).
\]
The early estimates \eqref{eq:early-coarse} and \eqref{eq:damping}
give a decaying bound for $\ell_v+2k$ on $\gamma$; a further
common width reduction gives $|\ell_v+2k|<(k-h_c)/2$.
The scalar and conformal bounds close jointly by first contact.
Retaining the decaying value on $\gamma$ in the homogeneous
term in the $\lambda$ bound also gives the time decay in
\eqref{eq:late-coarse}.
In particular the horizontal integral needed below is
\begin{equation}\label{eq:horizontal}
 \int_{u_\gamma(v)}^U|\phi_u|\dd u
 \le CM_p U^b+\frac{CU^{1-A}}{1-A}.
\end{equation}
These estimates close in the stated initial data spaces, including
the endpoint $p=4$ with its $L^{p/2}$ geometric derivative.

\textit{Exact conformal exponent.}
Write $y=r-r_-$ and $\Delta m=\varpi-M$ in the early region.
On a fixed root neighborhood $d(r)=-F_0(r)/(r-r_-)$ extends
continuously with $d(r_-)=2k$ and has a positive lower bound.
Thus
\begin{equation}\label{eq:damping}
 y_v=-\kappa d(r)y-2\kappa\Delta m/r,\qquad
 \int_{s_u}^v|y(w)|\dd w\le C_0(Y-r_-)+C_Ye^{-h_gs_u}.
\end{equation}
Indeed compare $|y|$ with the convolution of
$C_Ye^{-h_gw}$ and $e^{-c_0(v-w)}$, and exchange the two
integrals. This works even at equal decay rates and when $r$
crosses $r_-$. The Lipschitz bound for $K(r,M)$ gives the same
bound for $\int|K+k|$. Integrating
$\ell_v+2k=(\ell_v-2K)+2(K+k)$ yields
\begin{equation}\label{eq:exact-early}
 \Omega^2(u,v)=\Omega_Y^2(u)e^{-2k(v-s_u)+R_\Omega(u,v)},
 \qquad |R_\Omega|\le C_0(Y-r_-)+C_Ye^{-h_gs_u}.
\end{equation}
The constant $C_0$ is chosen on the root neighborhood before $Y$.

For signed data the late conformal equation is estimated absolutely.
On $\gamma$, the incoming constraint gives
$\Omega^2\le4(-\nu)\le Cu^{-1+\delta_\nu}$.
Set $\eta_\Omega=a\delta_\nu/(2\alpha)>0$.
Since $\eta_\Omega<k+h_c$, integrating the closed bound
$\ell_v\le-(k+h_c)$ gives
\[
 \Omega^2(u,v)\le Cu^{-1+\delta_\nu}
 e^{-(k+h_c)(v-v_\gamma(u))}
 \le Cu^{-1+\delta_\nu/2}e^{-\eta_\Omega v}.
\]
The second inequality uses the upper arrival scale
$u e^{av_\gamma(u)/\alpha}\le C$.
Half of the positive incoming power therefore suffices for
uniform time decay, while the other half remains integrable:
\begin{equation}\label{eq:omega-integral}
 \int_{u_\gamma(v)}^U\Omega^2\dd u
 \le \frac{CU^{\delta_\nu/2}}{\delta_\nu}e^{-\eta_\Omega v}.
\end{equation}
Equations \eqref{eq:late-coarse}--\eqref{eq:horizontal} bound the
other two horizontal absolute integrals in \eqref{eq:wave} by
$Ce^{-h_cv}$ and $Ce^{-\eta_\lambda v}$.
The pointwise damping bound for $y$ at $v=\alpha s_u$, together
with \eqref{eq:early-coarse}, gives
$|\ell_v+2k|(u_\gamma(v),v)\le C_Ye^{-\eta_\gamma v}$
for some $\eta_\gamma>0$. Horizontal integration therefore yields
\begin{equation}\label{eq:late-conformal}
 |\ell_v(u,v)+2k|\le Ce^{-\eta v},\qquad
 u_\gamma(v)\le u\le U,\qquad \eta>0,
\end{equation}
uniformly throughout the late region, using the partial $v$ derivative
on its matching boundary. Now $\Omega_Y^2(u)\asymp u^{-1}$ and
$e^{2ks_u}\asymp u^{-\rho}$ by \eqref{eq:arrival} and the fixed-band
bounds. Integrate \eqref{eq:late-conformal} from $\alpha s_u$ and
use \eqref{eq:exact-early} at that lower endpoint. This gives
\begin{equation}\label{eq:corner}
 D_Y\asymp u,\qquad
 \Omega^2(u,v)\asymp u^{-1-\rho}e^{-2kv},\qquad v\ge v_\gamma(u).
\end{equation}
Both the time exponent and the corner power are exact. The comparison
constants may depend on the fixed construction radius $Y$.

\textit{Radiation on finite time intervals.}
We now improve the coarse decay to the critical energy bound.
All kernels in this step act on absolute values, so the same proof
applies to signed data. The initial horizontal source is
$u^{-b}\int_0^u|f|=H_b|B|(\tau)$, with $L^2$ norm at most
$\|B\|_2$. At a fixed radius its time norm is controlled by
\eqref{eq:arrival}. In the outer neighborhood, conjugation by
$e^{kv}$ gives the feedback kernel
$\epsilon_R e^{-(a-k)(v-w)}\one_{w\le v}$ with
$\epsilon_R/(a-k)<1/2$.
In particular the vertical increment at $x=u_R(t)$ satisfies
\begin{equation}\label{eq:redshift}
 xe^{kt}|\zeta(x,t)-\zeta(x,0)|
 \le C_*e^{-(a-k)t}
 +C_R\int_{v_*}^t e^{-(a-k)(t-w)}M_R(w)\dd w,
\end{equation}
where $M_R=e^{kv}\sup_{u\le u_R(v)}|\theta|$.
The rate follows from $xe^{at}\asymp1$ and is $a-k=a/p>0$.
The finite-history coefficient is bounded by $C\mathcal B$.
The entire finite band can be treated with one spatial weight.
Choose a common late time $T$ on the arrival curve and set
$w(v)=u_Y(\max\{T,v\})$. This boundary is nonincreasing and
$w(v)e^{av}\le\Delta_Y$ for a fixed $\Delta_Y$.
Every past rectangle below it has $r\ge Y$; the coarse band
estimates give $|\nu|\le Ne^{av}$ and $|\lambda|\le L$ there.
The constants are uniform over the data family and in the terminal time.
For $\mu>0$, define
\[
 M_\mu(v)=\sup_{0\le u\le\min(U,w(v))}
       e^{kv-\mu ue^{av}}|\theta(u,v)|,\qquad
 (QM)(v)=\int_0^v e^{-(a-k)(v-t)}M(t)\dd t.
\]
In the vertical scalar integral, $e^{\mu ue^{at}}\le e^{\mu ue^{av}}$
for $t\le v$. Retaining this spatial factor in the horizontal
integral gives
\[
 e^{-\mu ue^{av}}e^{av}\int_0^u e^{\mu xe^{av}}\dd x
 =\frac{1-e^{-\mu ue^{av}}}{\mu}\le\frac1\mu.
\]
Thus the actual transport equations imply
\begin{align}\label{eq:band-feedback}
 M_\mu(v)&\le F(v)+\frac{LN}{\mu Y^2}(QM_\mu)(v),\nonumber\\
 F(v)&=e^{kv}|h(v)|+\frac{Lr_\bullet}{Y}e^{kv}
       \int_0^{\min(U,\Delta_Ye^{-av})}|f(x)|\dd x.
\end{align}
All integration histories remain below $w$ by its monotonicity.
The identity $(QM)'=M-(a-k)QM$, with $QM(0)=0$, and
Cauchy--Schwarz give $\|QM\|_{L^2(0,V)}\le(a-k)^{-1}\|M\|_{L^2(0,V)}$.
Choose $\mu=1+4LN/[Y^2(a-k)]$ to absorb the feedback.
The incoming part of $F$ has squared time norm at most
$(Lr_\bullet/Y)^2\Delta_Y^{2b}\|B\|_2^2/a$, by the exponential
average above and $k=ab$; the outgoing part is integrable since
$s>k$. It follows that $\|M_\mu\|_2\le C_Y\mathcal B$.
Since $e^{\mu ue^{av}}\le e^{\mu\Delta_Y}$ on the domain,
the same bound holds for $M_0$, uniformly in $V$.

The vertical integral also gives, on $u=u_Y(v)$ for $v\ge T$,
\[
 ue^{kv}|\zeta(u,v)-\zeta(u,0)|
 \le \frac{N\Delta_Y}{Y}e^{\mu\Delta_Y}(QM_\mu)(v).
\]
The direct transverse trace
$u_Y(t)e^{kt}(r_\bullet-u_Y(t))f(u_Y(t))$ is bounded in $L^2_t$
by $C_Y\|B\|_2$ using the Jacobian in \eqref{eq:arrival}.
Adding the direct term and the vertical increment bounds the total traces
\begin{equation}\label{eq:traces}
 \Theta_Y(t)=e^{kt}\theta(u_Y(t),t),\qquad
 Z_Y(t)=u_Y(t)e^{kt}\zeta(u_Y(t),t)
\end{equation}
and $M_0$ on the band for $v\ge T$. At positive $u$, continuity of the incoming
profile and the horizontal scalar integral identify this transverse
trace with the actual partial derivative $r\phi_u$.

In the early region the exact factors are
\begin{align}\label{eq:early-factors}
 -\lambda(x,v)&=\Lambda_Y(t)e^{-2k(v-t)+R_\Omega(x,v)+J_Y(x,v)},
 \quad t=v_Y(x),\quad\Lambda_Y=-\lambda_Y,\nonumber\\
 -\nu(x,w)&=(-\nu_Y(t))\frac{\kappa_Y(t)}{\kappa(x,w)}
                          e^{-2k(w-t)+R_\Omega(x,w)}.
\end{align}
Choose $Y$ near $r_-$ and the common initial width small so that
$\Lambda_Y\le\Lambda_*\le C(Y-r_-)$ and the early errors are
bounded by one. Put $v_0=v_Y(U)\ge T$ and $l=\max(v_0,v/\alpha)$.
For $v\ge v_0$, the horizontal source from $Z_Y$ is bounded by
\[
 C\Lambda_*\int_l^v e^{-k(v-t)}|Z_Y(t)|\dd t.
\]
Set $M_e(v)=e^{kv}\sup_{0\le u\le\min(U,w(v/\alpha))}|\theta(u,v)|$.
This continuous envelope equals $M_0$ for $v\le v_0$, contains $M_0$
thereafter and controls every vertical history by monotonicity of $w$.
Its boundary source is
bounded by $M_0(v)$. Substituting the vertical increment gives the feedback
\begin{equation}\label{eq:early-kernel}
 C\Lambda_*^2\int_l^v\int_t^v
       e^{-k(v-t)}e^{-3k(w-t)}M_e(w)\dd w\dd t.
\end{equation}
The exact Jacobian in \eqref{eq:arrival} cancels $-\nu_Y$.
For fixed $v$, set
\[
 P_v(t)=e^{-k(v-t)}\int_t^v e^{-3k(w-t)}M_e(w)\dd w.
\]
Differentiation gives $P_v'=4kP_v-e^{-k(v-t)}M_e(t)$ and $P_v(v)=0$.
The finite-interval identity is
\[
 4k\int_l^v P_v(t)\dd t
 =\int_l^v e^{-k(v-t)}M_e(t)\dd t-P_v(l)
 \le\int_{v_0}^v e^{-k(v-t)}M_e(t)\dd t.
\]
The resulting $L^2(v_0,V)$ operator bound is $1/(4k^2)$, so choose $Y$ to make the feedback less
than $1/4$. The constants in this small coefficient depend only
on the fixed root neighborhood, not on the finite-band constants.
The total trace on the matching curve satisfies
\begin{equation}\label{eq:z-gamma}
 |Z_\gamma(t)|\le |Z_Y(t)|+
 C_Y\int_t^{\alpha t}e^{-3k(w-t)}M_e(w)\dd w,
 \quad Z_\gamma(t)=u_Y(t)e^{kt}\zeta(u_Y(t),\alpha t).
\end{equation}
Positivity bounds this integral by
$R_V(t)=\int_t^V e^{-3k(w-t)}M_e(w)\dd w$ whenever $\alpha t\le V$.
Time reflection turns $R_V$ into the causal kernel with rate $3k$;
the same energy estimate gives
$\|R_V\|_{L^2(0,V)}\le(3k)^{-1}\|M_e\|_{L^2(0,V)}$.
Together with the matching trace bound, this controls $Z_\gamma$
on $v_0\le t\le V/\alpha$, uniformly in $V$.

In the late region first assume $J_Y\le1/4$ up to a finite terminal
time. Equations \eqref{eq:corner} and $D=D_Ye^{J_Y}$ give
$-\lambda\le Cx^{-\rho}e^{-2kv}$.
Using the total trace in \eqref{eq:z-gamma}, changing variables
$w=\alpha v_Y(x)$, and $x\asymp e^{-aw/\alpha}$ bounds the
free horizontal source by
\begin{equation}\label{eq:late-free}
 C\int_{v_\gamma(U)}^v e^{-k(v-w)}e^{-\xi w}
               |Z_\gamma(w/\alpha)|\dd w,
 \qquad \xi=k-\frac{a(\rho-b)}\alpha
             =\frac{k\beta}{\alpha}>0.
\end{equation}
The vertical increment and
$-\nu\le Cx^{-1+\delta_\nu}$ give
\begin{align}\label{eq:late-feedback}
 &Ce^{-kv}\int_{u_\gamma(v)}^U x^{-1-\rho+\delta_\nu}
                \int_{v_\gamma(x)}^v e^{-kw}M_l(w)\dd w\dd x
 \nonumber\\
 &\hspace{12mm}\le C\int_{v_\gamma(U)}^v
       e^{-k(v-w)}e^{-\chi w}M_l(w)\dd w,
 \qquad \chi=2k-\frac{a(\rho-\delta_\nu)}\alpha=3\xi>0,
\end{align}
where $M_l(v)=e^{kv}\sup_{u_\gamma(v)\le x\le U}|\theta(x,v)|$.
To obtain the second line, use the same arrival change of variables
and put $\mu=2k-\chi>0$. The primitive
\[
 P_v(t)=e^{\mu t}\int_t^v e^{-kw}M_l(w)\dd w,
 \qquad P_v'=\mu P_v-e^{(k-\chi)t}M_l(t)
\]
vanishes at $t=v$. Integrating from $v_\gamma(U)$ to $v$ and
discarding the nonnegative lower-endpoint value bounds
$e^{-kv}\int P_v$ by the causal integral in \eqref{eq:late-feedback},
with a factor $1/\mu$ absorbed in $C$.
The causal exponential estimate gives operator norm at most
$Ce^{-\chi v_\gamma(U)}/k$.
Choose $U$ last to make this less than $1/4$.
The time dilation in \eqref{eq:late-free} has norm $\sqrt\alpha$.
Finite-time absorption bounds $M_e$ and $M_l$ by $C\mathcal B$
in $L^2_v$, uniformly on their entire entered regions. It proves
\eqref{eq:energy-bound} under the bootstrap, with constants
independent of its terminal time or the positive ray $u$.

\textit{Flux closure and boundary geometry.}
Since $J_Y=\log(D/D_Y)$, \eqref{eq:flux} and the lower bound in
\eqref{eq:corner} imply
\begin{equation}\label{eq:closure}
 D(u,v)-D_\gamma(u)\le Cu^{1+\rho}\mathcal B^2,\qquad
 J_Y(u,v)\le b_\gamma(u)+\log(1+Cu^\rho\mathcal B^2).
\end{equation}
Shrink the same $U$ so that each term on the right is at most
$1/16$. On a finite terminal rectangle the entered interval has a
positive lower endpoint, and newly entering points have $J_Y=0$.
A first point with $J_Y=1/4$ is therefore impossible. The positive
radius bound and finite-rectangle continuation extend the solution
to every finite time. Monotone convergence gives the energy bound.

The mass bound follows directly from $F=4\nu D$:
\begin{equation}\label{eq:algebraic-mass}
 \varpi=\frac r2\left(1+\frac{e^2}{r^2}
                    -\frac{\Lambda r^2}{3}-4\nu D\right),
 \qquad D\le Cu,\qquad -\nu D\le Cu^{\delta_\nu}.
\end{equation}
The last two bounds hold beyond $\gamma$ by $D=D_Ye^{J_Y}$,
\eqref{eq:closure}, and \eqref{eq:late-coarse}. The positive compact
radius range gives a uniform upper bound for the late mass.
The finite-rectangle estimates above bound its initial values
uniformly. Since $\partial_v\varpi\ge0$, the lower bound propagates
forward and the late upper bound controls every earlier point on
the same ray. Thus $\varpi$ is uniformly bounded on the whole strip.
This construction applies to every fixed finite $E_*$, with $U$
chosen according to its size.

Finally the integrable estimates for $\lambda$ and $\ell_v+2k$
give, for some $\eta_I>0$,
\begin{equation}\label{eq:tails-geometry}
 \sup_{v\ge T}\|r_n(v)-r_n(T)\|_{C(I)}+
 \sup_{v\ge T}\|\log\cA_n(v)-\log\cA_n(T)\|_{C(I)}
 \le C_Ie^{-\eta_I T}.
\end{equation}
Equation \eqref{eq:flux} gives the common $e^{-2kT}$ tail of
$r_n\nu_n$. Local continuous dependence from \eqref{eq:local}
and these tails prove the asserted compactified convergence.
The common tails follow from coarse scalar decay and the horizontal
$L^1$ bound, and remain valid as the radiation energy moves to late times.
\end{proof}

%% file: sections/radiation.tex
\section{A radial mode and the full-time radiation limit}\label{sec:radiation}
Set
\begin{equation}\label{eq:amplitude-data}
 c_\nu=\exp\int_0^\infty2(K_H-K_+)\dd v,\qquad
 C_h=\frac{r_\bullet}{r_+c_\nu^b},\qquad d_h=\log(c_\nu u_*).
\end{equation}
The radial mode constructed below depends only on $F_0$ and $k$.
The radiating reference solution depends also on the fixed tail $h$.
The comparison uses the slow variation of $B_n$: multiplying the
radial mode by this envelope leaves a transport error controlled by
$\|B_n'\|_2$. An energy estimate on the evolving metric controls
that error over the full time interval. The endpoint bounds for the
mode allow the radius cutoffs to be removed, after which transverse
transport carries the comparison to each positive observation ray.

\begin{proposition}\label{prop:radiation}
There is a finite positive number $\mathcal T$ such that, on the
common region of Proposition~\ref{prop:common}, for each compact
$I\subset(0,U]$ and a fixed sufficiently late $v_I$,
\begin{equation}\label{eq:radiation-limit}
 \left\|e^{kv}(\theta_n-\theta_0)-C_h\mathcal T B_n(av)
             \right\|_{L^2([v_I,\infty);C^0(I))}\longrightarrow0.
\end{equation}
This holds for a single packet and for a fixed finite signed sum.
The parameters $t_n,L_n$ of each component may diverge independently.
\end{proposition}

\begin{proof}
\textit{An integrable positive system.}
Let $z>0$ solve $z'/z=a/F_0$ with $z/(r_+-r)\to1$ at $r_+$,
and put $q_0=z^b$. In the inward variable $x=r_+-r$ solve
\begin{equation}\label{eq:volterra}
 X_x=\frac{a q_0 V}{r},\qquad V_x=\frac{X}{arq_0},\qquad
 (X,V)|_{x=0}=(0,r_+b),\qquad W=q_0V.
\end{equation}
The fourth root is $r_o=-(r_-+r_++r_c)$, and
$F_0(r)=-\Lambda\prod_{j\in\{o,-,+,c\}}(r-r_j)/(3r^2)$.
Partial fractions give the normalized solution explicitly:
\begin{equation}\label{eq:radial-coordinate}
 z(r)=(r_+-r)\exp\sum_{j\in\{o,-,c\}}
 \frac{a}{F_0'(r_j)}
 \left[\log|r-r_j|-\log|r_+-r_j|\right].
\end{equation}
The inner exponent is $-a/(2k)$. Thus $q_0=z^b$ satisfies
$q_0\sim(r_+-r)^b$ at $r_+$ and
$q_0\asymp(r-r_-)^{-1/2}$ at $r_-$.
Both $q_0$ and $q_0^{-1}$ are integrable at both endpoints because
$0<b<1$. Set $H(x)=\int_0^x[aq_0/r+(arq_0)^{-1}]\dd y$.
The Volterra map is a contraction with constant $1/2$ in the norm
$\sup_x e^{-2H(x)}\max(|X(x)|,|V(x)|)$, since
$\int_0^x H'e^{2H}\dd y=(e^{2H(x)}-1)/2$.
It gives a unique continuous, hence bounded, pair on the closed
interval. Iteration from zero preserves nonnegativity; the integral
equations then give $V\ge r_+b>0$ and $X>0$ inward from the outer
endpoint. In particular
\begin{align}\label{eq:mode-ends}
 &X(r)=\frac{k}{b+1}(r_+-r)^{b+1}[1+O(r_+-r)],\quad
 W(r)=r_+b(r_+-r)^b[1+O(r_+-r)],\nonumber\\
 &\mathcal T:=X(r_-)\in(0,\infty),\quad
 X(Y)=\mathcal T+O(\sqrt{Y-r_-}),\nonumber\\
 &W(Y)=O((Y-r_-)^{-1/2}),\quad
 (-F_0(Y))W(Y)=O(\sqrt{Y-r_-}).
\end{align}
The inner remainder follows by integrating $X_x=a q_0V/r$
over the last interval. Also
\begin{equation}\label{eq:mode}
 X_x=\frac ar W,\qquad
 W_x=\frac{k}{-F_0}W+\frac{1}{ar}X.
\end{equation}
These two equations retain the complete scalar feedback and determine
the positive transmission amplitude.

\textit{Matching the actual geometry on a fixed band.}
Fix $r_-<Y<R<r_+$. At late times the entire band lies in the
common development. Uniformly in $n$, its coefficients obey
\begin{equation}\label{eq:geom-match}
 \kappa\to1,\quad\varpi\to M,\quad
 \frac{u\nu}{F_0}\to\frac1a,\quad
 \frac{z(r)}{c_\nu u e^{av}}\to1.
\end{equation}
We record the identities that give the normalization and the
Jacobian control needed for a time $L^2$ estimate. For
$c=-1/\lambda$, differentiation at fixed $r$ or fixed $u$ gives
\begin{align}
 c_v|_r&=\frac{\kappa_v|_u/\kappa}{\lambda}
                    -\frac{\theta^2}{r\lambda^2},\label{eq:cv}\\
 \frac{\kappa_v|_u}{\kappa}
 &=\int_0^u\left[-\frac{2\zeta\theta}{r^2}
       -\frac{\zeta^2\lambda}{r^2\nu}
       -\frac{2\kappa K\zeta^2}{r\nu}\right]\dd x,\label{eq:kv}\\
 \partial_v|_u\log\frac{u\nu}{F_0}
 &=2\kappa(\varpi-M)
            \left[\frac1{r^2}+\frac{F_0'}{rF_0}\right],\label{eq:ratio}\\
 \partial_v|_u[\log z(r)-av]
 &=a(\kappa-1)-\frac{2a\kappa(\varpi-M)}{rF_0}.\label{eq:phase}
\end{align}
Equation \eqref{eq:kv} follows by differentiating the integrated
constraint.

Here is a direct estimate of the errors in \eqref{eq:geom-match}.
The fixed band has a finite crossing time $D_{R,Y}$, since
$|F_0|$ has a positive lower bound there.
Equations \eqref{eq:outer-coarse}, \eqref{eq:kappa-bound}, and
finite layer integration give
\begin{equation}\label{eq:band-bounds}
 |\theta|\le C_{R,Y}e^{-h_cv},\quad
 |\zeta|\le C_{R,Y}(|f_n|+e^{(a-h_c)v}),\quad
 |\varpi-M|+|\log\kappa|+|\kappa_v/\kappa|+|c_v|\to0.
\end{equation}
The quadratic terms use \eqref{eq:data-bounds} or
$|f_n|\le bA_*u^{-1/p}$ and $u\asymp e^{-av}$.
To fix the constants in \eqref{eq:ratio}--\eqref{eq:phase},
choose a small fixed $\delta>0$ and meet the ray at
$w=a^{-1}\log(\delta/u)$. The outer estimates and
$-\nu_H\sim c_\nu e^{aw}$ give, uniformly on $ue^{aw}=\delta$,
\[
 \frac{-\nu(u,w)}{c_\nu e^{aw}}=1+O(\delta)+o_{w\to\infty}(1),\qquad
 \frac{r_+-r(u,w)}{c_\nu\delta}=1+O(\delta)+o_{w\to\infty}(1).
\]
The second identity follows by integrating $r_u=\nu$ from the
horizon to the chosen ray; the horizon radius error divided by the fixed
$c_\nu\delta$ tends to zero. Since $F_0(r)\sim a(r-r_+)$ and
$z(r)\sim r_+-r$, the sum of the two logarithmic normalization
errors at $w$ is at most $\omega(\delta)+o_{w\to\infty}(1)$,
where $\omega(\delta)\to0$.

The arrival bounds give fixed $j,D>0$ with $j\le ue^{av}\le D$
throughout the band. Choose $\delta<j$ and
$2c_\nu\delta<r_+-R$. Then
$0<v-w\le a^{-1}\log(D/\delta)=:L_\delta$.
Before the ray reaches the inner edge of the band, the mass and lapse
estimates give $|\varpi-M|+|\kappa-1|\le Ce^{-\gamma t}$ for
some $\gamma>0$. At sufficiently late $w$, the section radius lies
between $R$ and $Z_\delta=r_+-c_\nu\delta/4$, and
$F(r,\varpi)<0$ there. The outgoing constraint preserves
$\lambda<0$, so the segment to $v$ stays in $[Y,Z_\delta]$.
Integrating \eqref{eq:ratio}--\eqref{eq:phase} on this compact
interval bounds the additional normalization error by
$C_\delta L_\delta e^{-\gamma w}$. First let $v\to\infty$
at fixed $\delta$, and then let $\delta\downarrow0$.
This proves \eqref{eq:geom-match}, with all cutoffs chosen before
the incoming profile. Together with \eqref{eq:band-bounds}, it
gives a bound $\epsilon_{R,Y}(T)$ for the normalized coefficient
errors and $c_v|_r$ on $v\ge T$, satisfying
\begin{equation}\label{eq:eps}
 \epsilon_{R,Y}(T)\longrightarrow0\qquad(T\to\infty).
\end{equation}
The arrival phase $\tau(u_R(v))=av+d_h-\log z(R)+e_R(v)$
has $e_R,e_R'\to0$ as well, since
$-u_R'/u_R=\lambda/(u_R\nu)$.

\textit{The energy estimate for a slow envelope.}
Use the total scalar field on each actual metric. In the band
coordinate $x=R-r$ put
\begin{equation}\label{eq:band-variables}
 \Theta=e^{kv}\theta,\quad Z=e^{kv}u\zeta,\quad
 A=\frac{\lambda}{ru\nu},\quad B_c=\frac{u\nu}{r\lambda},
 \quad c=-\lambda^{-1}.
\end{equation}
The equations are exactly
\begin{equation}\label{eq:band-system}
 \Theta_x=AZ,\qquad Z_x+cZ_v=kcZ+B_c\Theta.
\end{equation}
For the limiting geometry $A_0=a/r$, $B_0=1/(ar)$,
$c_0=(-F_0)^{-1}$. The same rescaling as in \eqref{eq:volterra},
$\widehat Z=Z/q_0$, removes the singular reference diagonal:
\begin{equation}\label{eq:rescaled-band}
 \Theta_x=Aq_0\widehat Z,\qquad
 \widehat Z_x+c\widehat Z_v=(B_c/q_0)\Theta+k\delta c\widehat Z,
 \qquad \delta c=c-c_0.
\end{equation}
Compare with $(X,V)H_n(v)$, where
$H_n=C_hB_n(av+d_R)$ and $d_R=d_h-\log z(R)$, and
write the errors as $p_1=\Theta-XH_n$, $q_1=\widehat Z-VH_n$.
They satisfy the same system with the added sources
\begin{equation}\label{eq:residual}
 S_1=(A-A_0)q_0VH_n,\qquad
 S_2=\left(\frac{B_c-B_0}{q_0}X+k\delta c V\right)H_n
                      -acVC_h B_n'(av+d_R).
\end{equation}
Fix a late $T$ and set
$\mathcal N=\|p_1\|_{L^2(T,S)}^2+\|q_1\|_{L^2(T,S)}^2$.
Integration by parts gives
\begin{align}\label{eq:energy-halfline}
 \mathcal N'={}&2\int_T^S(Aq_0+B_c/q_0)p_1q_1\dd v
                  +\int_T^S(2k\delta c+c_v)q_1^2\dd v\nonumber\\
 &+c(x,T)|q_1(x,T)|^2-c(x,S)|q_1(x,S)|^2\nonumber\\
 &+2\langle S_1,p_1\rangle+2\langle S_2,q_1\rangle.
\end{align}
Discard the nonpositive upper flux and use $2|p_1q_1|\le p_1^2+q_1^2$.
The reference growth has a finite integral on the entire radius interval:
\begin{equation}\label{eq:integrable-growth}
 G_0=\int_{r_-}^{r_+}\left(\frac{aq_0}{r}+\frac{1}{arq_0}\right)\dd r<\infty.
\end{equation}
Gronwall's inequality, followed by $S\to\infty$, yields
\begin{equation}\label{eq:band-energy}
 \mathcal N(x)^{1/2}\le C e^{G_0/2+\int_0^x e(y)\dd y}
 \left[\mathcal N(0)^{1/2}+\int_0^x(\|S_1\|_2+\|S_2\|_2)\dd y
       +\left(\int_0^x c(y,T)|q_1(y,T)|^2\dd y\right)^{1/2}\right],
\end{equation}
where $C$ is absolute and
$e=(q_0\|A-A_0\|_\infty+q_0^{-1}\|B_c-B_0\|_\infty)/2
+k\|\delta c\|_\infty+\|c_v\|_\infty/2$.
The positive bottom inflow is retained. After the coefficient errors
vanish on each fixed band, the factor $Ce^{G_0/2}$ is independent
of both radius cutoffs.

At fixed $R,Y,T$, the bottom term converges to that of the
zero-incoming reference. Indeed the packet vanishes on the fixed
history, and finite-time continuous dependence gives convergence
of the geometry and its radius coordinates. On this positive
$u$ interval, $f_n$ also vanishes eventually, so the vertical
scalar integral gives convergence of the transverse trace.
Denote the reference quantities by $c_h$ and $\widehat Z_h$.
Then
\begin{equation}\label{eq:bottom-limit}
 \int_0^{R-Y}c_n(x,T)|q_{1,n}(x,T)|^2\dd x
 \longrightarrow
 \int_0^{R-Y}c_h(x,T)|\widehat Z_h(x,T)|^2\dd x
 \le C_{R,Y}e^{-2\gamma T},\qquad \gamma>0.
\end{equation}
To obtain the last bound, choose $k<h_*<\min(a,s)$.
For $f=0$ the scalar integral estimates before arrival at the inner
band boundary have no incoming H\"older term. The same spatial Gronwall
argument as in the common construction, using $ue^{av}\le D_Y$,
gives $|\phi_v|\le Ce^{-h_*v}$ and
$|\zeta|\le Ce^{(a-h_*)v}$.
Thus $|\Theta_h|+|\widehat Z_h|\le C_{R,Y}e^{-\gamma v}$
with $\gamma=h_*-k$, since $u\asymp e^{-av}$ and $q_0$ is
bounded away from zero on the fixed band. Also $-\lambda_h$
is bounded away from zero there. This proves
\eqref{eq:bottom-limit}; its reference contribution is removed
by taking $T\to\infty$ after the sequence limit.
The source term is bounded by
$C_{R,Y}(\epsilon_{R,Y}(T)\sqrt{E_*}+\|B_n'\|_2)$.

It remains to specify the outer data in the same normalization.
Put $\delta_+=r_+-R$ and
$(J_cB)(t)=c\int_0^\infty e^{-cz}B(t+z)\dd z$.
The leading horizontal integral in the thin outer region is
$a c_\nu(r_\bullet/r_+)b e^{(a+k)v}
 \int_0^{u_R(v)}x^b B_n(\tau(x))\dd x$.
The substitution $x=u_R e^{-z}$ gives $J_{b+1}$ and
the leading coefficient $C_h k\delta_+^{b+1}/(b+1)$ in
\eqref{eq:mode-ends}. The direct vertical trace gives
$C_h r_+b\delta_+^b$. The relative radius and coefficient errors
tend to zero. After division by $q_0(R)$, the transverse feedback
has a squared error bounded by a constant times
$\delta_+^{2(1-b)}(1+E_*)$, which tends to zero because $b<1$.
The prescribed outgoing tail is retained in the transport integrals;
its conjugated late energy tends to zero since $s>k$.
The earlier vertical history is controlled by the redshift kernel
\eqref{eq:redshift}, whose rate is $a-k>0$.
The phase error $e_R(v)$ in \eqref{eq:geom-match} has a bounded
time Jacobian. For its action on a packet the fundamental theorem
of calculus bounds the $L^2$ error by
$C\|e_R\|_\infty\|B_n'\|_2$.
Finally $\|J_{b+1}B_n-B_n\|_2\le(b+1)^{-1}\|B_n'\|_2$.
Together these bounds give the following useful form of the outer
estimate. For any $\varepsilon>0$, one can choose $R$ arbitrarily
close to $r_+$ and then a common late $T$ such that
\begin{equation}\label{eq:outer-traces}
 \|\Theta_{R,n}-X(R)H_n\|_{L^2(T,\infty)}^2
 +\|\widehat Z_{R,n}-V(R)H_n\|_{L^2(T,\infty)}^2
 \le\varepsilon(1+E_*)+C_{\varepsilon,R}\|B_n'\|_2^2.
\end{equation}
The additive term includes outgoing radiation. The constants
are uniform in $n$ and the finite terminal time.

Apply \eqref{eq:band-energy} at fixed $Y,R$ and take the sequence
limit, followed by $T\to\infty$. Equations
\eqref{eq:bottom-limit} and \eqref{eq:eps} remove the bottom and
coefficient errors. Letting the outer accuracy $\varepsilon$ tend
to zero, with $R\uparrow r_+$, gives an arbitrarily small total-field
error on a sufficiently late half-line at $Y$. The propagation
factor tends to $Ce^{G_0/2}$, independently of $R$.
Auxiliary estimates may use shorter incoming initial segments.
Uniqueness identifies the resulting solutions with restrictions of
the original development, and each fixed radius curve eventually
lies inside their domains.
The common $U$ is therefore fixed throughout this argument.

Subtract the reference trace, whose late $L^2$ norm tends to zero
by the same $h_*$ estimate. On any fixed history interval the
actual traces converge to their reference counterparts, while
the packet eventually vanishes. A fixed phase change costs at
most its length times $\|B_n'\|_2$, so $d_R$ can be replaced by
$d_h$. Combining the finite history and late tail proves
\begin{equation}\label{eq:band-limit}
 \|\Theta_{Y,n}-\Theta_{Y,0}-C_hX(Y)B_n(av+d_h)\|_{L^2_v}
 +\|\widehat Z_{Y,n}-\widehat Z_{Y,0}-C_hV(Y)B_n(av+d_h)\|_{L^2_v}
 \longrightarrow0.
\end{equation}
Here the time norm starts at any fixed common time at which the
radius curves exist, and each trace is evaluated on its own
development's $Y$-radius curve. The argument uses bounded envelope
energy, vanishing derivative energy and incoming $L^p$ norm, and
escape from fixed history intervals. No bound on $t_n/L_n$ enters.

\textit{Removing the inner cutoff.}
Denote the construction radius of Proposition~\ref{prop:common} by
$Y_0$ and, for each development, its already chosen matching curve
by $\gamma_0$. Keep these curves and the common $U$ fixed.
Let an auxiliary
$Y=r_-+\delta_-$ approach the inner root.
At a sufficiently late time depending on this fixed $Y$, the
region between $Y$ and $\gamma_0$ is contained in the original
early region. The damping argument \eqref{eq:damping}, integrated
from this new radius, and \eqref{eq:arrival} give
$\Lambda_Y\le L\delta_-$ and $-u_Y'/u_Y\le2a$, with $L$
chosen on the fixed root neighborhood.
Write $\mathcal R_n(u,v)=e^{kv}\theta_n(u,v)$,
$\widehat Z_{Y,n}=e^{kv}u_Y\zeta_n/q_0(Y)$, and
$\mathcal D_{Y,n}(v)=\sup_{u_Y(v)\le u\le u_{\gamma_0}(v)}
|\mathcal R_n(u,v)-\mathcal R_n(u_Y(v),v)|$.
The fixed-band argument with the actual reference traces subtracted
gives the mode coefficients $C_hX(Y)$ and $C_hV(Y)$; the reference
traces are in $L^2$. Put $r_{\min}=r_-/2$, $N=2\exp(2)$,
$C_Y=\exp(33/16)(L\delta_-)^2$, and
$n_Y=\exp(-1)[-F_0(Y)]/a$.
The restarted damping and the same two kernels give
\begin{align}\label{eq:total-inner-difference}
 \int_{b_0}^\infty\mathcal D_{Y,n}^2\dd v
 &\le 2\left(\frac{C_Yq_0(Y)}{r_{\min}k n_Y}\right)^2
       \int_l^\infty\widehat Z_{Y,n}^2\dd v
   +2\left(\frac{C_YN}{4k^2r_{\min}^2}\right)^2 E_{\rm tot},
\end{align}
where $E_{\rm tot}$ bounds $\int_0^\infty\sup_{0\le u\le U}
|\mathcal R_n(u,v)|^2\dd v$ and $b_0/\alpha\ge l$ is sufficiently
late at this fixed $Y$. The source and observation cutoffs are
independent. First take the sequence limit and remove the reference
tail by increasing $l$. Since
$[\delta_-^2q_0(Y)/(-F_0(Y))]^2=O(\delta_-)$ and $V$ is bounded,
the right side has upper limit at most
$C\delta_-E_*+C\delta_-^4E_{\rm tot}$.
Continuity of $X$ at $r_-$ then yields
\begin{equation}\label{eq:gamma-total-tail}
 \lim_{T\to\infty}\limsup_n\int_T^\infty
 |\mathcal R_n(u_{\gamma_0}(v),v)-C_h\mathcal T B_n(av)|^2\dd v=0.
\end{equation}
The time cutoff may depend on $Y$, while $U$ and $\gamma_0$ stay fixed.

\textit{Propagation to a fixed ray.}
Let $\mathcal D_n(v)$ be the supremum of
$|\mathcal R_n(u,v)-\mathcal R_n(u_{\gamma_0}(v),v)|$
over $u_{\gamma_0}(v)\le u\le U$, and set it to zero before
the curve enters this interval. Retain the signed horizontal
increment when applying \eqref{eq:late-free}--\eqref{eq:late-feedback}.
With $A_n(w)=|Z_{\gamma,n}(w/\alpha)|$, for
$v\ge v_{\gamma_0}(U)$ this gives
\[
 \mathcal D_n(v)\le
 C_A\int_{v_{\gamma_0}(U)}^v e^{-k(v-w)}e^{-\xi w}A_n(w)\dd w
 +C_M\int_{v_{\gamma_0}(U)}^v e^{-k(v-w)}e^{-\chi w}M_{l,n}(w)\dd w.
\]
Let $E_\gamma$ and $E_l$ bound the squared $L^2$ norms of
$Z_{\gamma,n}$ and $M_{l,n}$ on their respective entered regions,
as established in Sec.~\ref{sec:common}, uniformly in the terminal
time and the profile. Time dilation then bounds the squared source
norms by $\alpha E_\gamma$ and $E_l$.
Since $0<\xi<k$ and $\chi=3\xi$, the identity
\[
 e^{\xi v}e^{-k(v-w)}e^{-dw}
 =e^{-(k-\xi)(v-w)}e^{-(d-\xi)w},\qquad d\in\{\xi,\chi\},
\]
and the causal $L^2$ bound give, for every finite $S\ge0$,
\begin{equation}\label{eq:late-difference-weighted}
 \int_0^S e^{2\xi v}\mathcal D_n(v)^2\dd v
 \le \frac{2C_A^2\alpha E_\gamma+2C_M^2E_l}{(k-\xi)^2}
 =:E_D.
\end{equation}
Before first entry the integrand vanishes. Exhaustion in $S$ thus
gives $\int_T^\infty\mathcal D_n^2\dd v\le E_De^{-2\xi T}$.
Only the difference receives the exponential weight; both sources
use the energy bounds already proved.

For a fixed compact $I\subset(0,U]$, the scale bound
$u_{\gamma_0}(v)\le D_r e^{-av/\alpha}$ places the curve below
$I$ at a time independent of $n$. Combining this difference tail
with \eqref{eq:gamma-total-tail} gives the same comparison on
$I$ in $L^2_vC(I)$. The reference radiation has a vanishing $L^2$
tail by \eqref{eq:energy-bound}. On every fixed history interval,
finite-time continuous dependence gives
$\mathcal R_n-\mathcal R_0\to0$ uniformly on $I$, while the
packet eventually vanishes there. First take $n\to\infty$ on
that interval, and then increase its endpoint. This proves
\eqref{eq:radiation-limit}, with the common region and energy
budget fixed throughout.
\end{proof}

%% file: sections/response.tex
\section{Mass accumulation and weak extension}\label{sec:response}
Write the zero-incoming boundary geometry as
\begin{equation}\label{eq:boundary-values}
 R_h(u)=r_0(u,\infty),\qquad N_h(u)=\nu_0(u,\infty),\qquad
 A_h(u)=\lim_{v\to\infty}e^{2kv}\Omega_0^2(u,v).
\end{equation}

\begin{proposition}\label{prop:mass}
On each positive compact interval $I\subset(0,U]$,
$R_h,A_h>0$ and $N_h<0$. For a single packet the mass difference
satisfies, uniformly for $u\in I$ and $v\in[v_I,\infty]$,
\begin{equation}\label{eq:mass-profile}
 \varpi_n(u,v)-\varpi_0(u,v)
 =\mathcal A_h(u)F_\psi((av-t_n)/L_n)+o_I(1),\qquad
 \mathcal A_h(u)=\frac{-2N_h(u)}{aA_h(u)}C_h^2\mathcal T^2E>0,
\end{equation}
where $F_\psi(y)=M_2^{-1}\int_{-\infty}^y\psi^2$.
For a finite signed sum $B_n$ the corresponding formula is
\begin{equation}\label{eq:sum-profile}
 \varpi_n(u,v)-\varpi_0(u,v)
 =\Gamma_h(u)\int_{-\infty}^{av} B_n(t)^2\dd t+o_I(1),
 \qquad \Gamma_h(u)=\frac{-2N_h(u)}{aA_h(u)}C_h^2\mathcal T^2.
\end{equation}
Each development, including the reference, has a local weak
extension across every horizon point $0<u<U$ with continuous
$W^{1,2}_{\rm loc}$ metric, $L^2_{\rm loc}$ connection, and
$H^1_{\rm loc}$ scalar.
The boundary is locally the future Cauchy horizon of the original
two characteristic initial hypersurfaces.
\end{proposition}

\begin{proof}[Proof of the mass laws.]
\textit{The positive geometric factor.}
Proposition~\ref{prop:common} gives the positive limits of $r_0$
and $\cA_0$. To establish $N_h<0$, it suffices to consider the
reference scalar. In the small nontrapped part the only negative
correction to $\theta_0=h$ is the horizontal integral with
$\lambda>0$. Equations \eqref{eq:lambda-h} and
\eqref{eq:outer-coarse} bound this correction by
$Ce^{-(2s+h_c)v}$.
Since $l<2s+h_c$, it is eventually smaller than $h/2$.
On the remaining fixed history interval, continuity and a sufficiently
small common width keep $\theta_0>0$. Vertical integration then
gives $\zeta_0(u,v)>0$ for $u,v>0$. In the trapped region both
scalar transports preserve these signs.

On a positive interval the bounded mass and \eqref{eq:flux} give
limits of $\zeta_0$ by its integrable vertical derivative.
Taking the boundary limit in the integrated transverse constraint
gives, for any $0<d<u$,
\begin{equation}\label{eq:strict-negative}
 \frac{N_h(u)}{A_h(u)}
 =\frac{N_h(d)}{A_h(d)}
 -\int_d^u\frac{\zeta_{0,\infty}(x)^2}{R_h(x)A_h(x)}\dd x<0.
\end{equation}
Here $N_h(d)\le0$ and the integrand is strictly positive. Choose
$d<\inf I$ to obtain a strictly negative bound on $I$.
This is the role of the nonzero outgoing tail in the mass response.

\textit{Strong convergence of the mass density.}
Set $\Theta_n=e^{kv}\theta_n$ and
$C_n(u,v)=-2\nu_n(u,v)/\cA_n(u,v)$.
The mass equation is $\varpi_{n,v}=C_n\Theta_n^2$.
The $C_n$ are uniformly bounded and converge to $C_0$ on the
compactified strip. Their late limits converge to
$C_\infty(u)=-2N_h(u)/A_h(u)$, with common tails.
Writing $\Delta\Theta_n=\Theta_n-\Theta_0$ gives the exact
decomposition
\begin{equation}\label{eq:density-decomposition}
 C_n\Theta_n^2-C_0\Theta_0^2
 =C_n(\Delta\Theta_n)^2+2C_n\Theta_0\Delta\Theta_n
                              +(C_n-C_0)\Theta_0^2.
\end{equation}
Proposition~\ref{prop:radiation} and the uniform $L^2_vC(I)$
bound imply
\begin{equation}\label{eq:square-limit}
 \| (\Delta\Theta_n)^2-C_h^2\mathcal T^2B_n(av)^2
                                    \|_{L^1_vC(I)}\to0.
\end{equation}
The product with $\Theta_0$ tends to zero after splitting at a
fixed $T$: use local convergence on $[v_I,T]$ and the fixed
reference $L^2$ tail on $[T,\infty)$. The last term of
\eqref{eq:density-decomposition} tends to zero by uniform
coefficient convergence. Finally replace $C_n$ by $C_\infty$ in
the leading square. On fixed histories $B_n(av)$ vanishes
eventually; on the late interval the coefficient difference is
bounded by $o_n(1)+C_Ie^{-\eta_I T}$. Since the total envelope
energy is bounded, this proves
\begin{equation}\label{eq:mass-density}
 \|\partial_v(\varpi_n-\varpi_0)
       -C_\infty C_h^2\mathcal T^2B_n(av)^2\|_{L^1_vC(I)}\to0.
\end{equation}
The starting mass difference at $v_I$ tends to zero. Integrating
an $L^1$ error controls every upper endpoint, including infinity.
The substitution $t=av$ proves \eqref{eq:sum-profile}, and
\eqref{eq:packets} gives \eqref{eq:mass-profile}.
The error tends to zero at fixed finite energy $E$.

\end{proof}

\subsection{Delay, duration, and two packets}
For a single packet, the terminal excess converges to the positive
$\mathcal A_h$ on $I$, whereas every fixed finite-time excess
tends to zero. This proves the two noncommuting limits in the
Letter. Define $v_{q,n}(u)$ as the first time $v\ge v_I$ at
which the excess reaches $q$ times its terminal excess.
Let $q$ range in a compact subinterval $[q_0,q_1]\subset(0,1)$.
Continuity of the excess and its positive terminal limit give
these first crossings for all sufficiently large $n$, uniformly
in $u$ and $q$. Evaluating \eqref{eq:mass-profile} there and using
terminal convergence gives
$F_\psi((av_{q,n}-t_n)/L_n)=q+o_I(1)$, since
$\min_I\mathcal A_h>0$.
As $\psi>0$ on $(0,1)$, $F_\psi$ maps $[0,1]$ continuously
and strictly increasingly onto $[0,1]$. The compact range of $q$
places $(av_{q,n}-t_n)/L_n$ in $(0,1)$ for large $n$.
Uniform continuity of the inverse gives, uniformly in $u$ and $q$,
\begin{equation}\label{eq:quantiles}
 v_{q,n}(u)=\frac{t_n+L_nF_\psi^{-1}(q)}a+o_I(L_n).
\end{equation}
Uniform approximation by $F_\psi$ locates the first crossing even
when the exact excess, a difference of two increasing masses, varies
nonmonotonically.

For two envelopes use $B_n=B_{1,n}+\sigma B_{2,n}$,
$\sigma=\pm1$, in \eqref{eq:data} with its full square.
The bound $\|B_n\|_2\le\sqrt{E_1}+\sqrt{E_2}$ supplies a
common energy budget. Its $L^p$ norm and derivative $L^2$ norm
tend to zero, and both supports escape every finite interval.
Thus \eqref{eq:sum-profile} proves the two-packet law.
At equal component energies $E_0$, common width $L_n$, and
separation $\delta L_n$, it becomes
\begin{equation}\label{eq:pair-shape}
 \frac{\varpi_n-\varpi_0}{2E_0\Gamma_h}
 =G_{\delta,\sigma}\left(\frac{av-t_n}{L_n}\right)+o_I(1),
 \qquad G_{\delta,\sigma}(y)=\frac1{2M_2}
 \int_{-\infty}^y[\psi(z)+\sigma\psi(z-\delta)]^2\dd z.
\end{equation}
The fixed component energies allow the energy of the combined
input to vary with overlap. For identical opposite packets the
total initial data equal the reference exactly, and uniqueness
gives zero response for every $n$.

\subsection{Weak extension and the original initial surface}\label{sec:extension}
\begin{proof}[Proof of the extension assertions in Proposition~\ref{prop:mass}.]
\textit{Regular coordinates.}
Fix one development and $I=[d,U]$ with $d>0$. Set
\begin{equation}\label{eq:regular-chart}
 a_U(v)=\Omega^2(U,v),\quad
 V(v)=-\int_v^\infty a_U(w)\dd w,\quad
 H(u,V)=\Omega^2(u,v)/a_U(v),\quad \widetilde\ell=\log H.
\end{equation}
Then $a_U\asymp e^{-2kv}$ and $H$ has a positive continuous
limit at $V=0$. For this one solution the radiation energy has
a vanishing tail. Since $D_v=\theta^2/(r\Omega^2)$,
$D$ has a continuous finite limit uniformly on $I$, bounded away
from zero; hence $q$ does too. Thus $r_V=-HD$ is continuous.
Bounded mass and the scalar equations give
\begin{align}\label{eq:extension-estimates}
 &|\zeta_v|\le C_I e^{-kv}Z(v),\qquad
 |\nu_v|\le C_Ie^{-2kv},\nonumber\\
 &|\ell_{uv}|\le C_I(e^{-kv}Z(v)+e^{-2kv}),
 \qquad Z(v)=e^{kv}\sup_{u\in I}|\theta(u,v)|\in L^2.
\end{align}
Here $\zeta$ is bounded on $I$ by vertical integration from a
finite time. Since each $f_n$ is smooth on $I$, the same estimates
give continuous limits of $\phi_u,\nu,\ell_u$.
Horizontal integration of the last bound controls
$\ell_v(u,v)-\ell_v(U,v)$, the numerator of
$\widetilde\ell_V$. For $\Psi=(r,\phi,\widetilde\ell)$,
\begin{equation}\label{eq:extension-norm}
 \Psi_u\in C^0,\qquad
 \int_{V(T)}^0\|\Psi_V(\cdot,V)\|_{C(I)}^2\dd V
 \le C_I\int_T^\infty(Z(v)^2+e^{-2kv})\dd v<\infty.
\end{equation}

To construct an extension near $u_0$, choose $u_a<u_0$ nearby.
Keep the actual data on $u=u_a$, $V<0$. For $V\ge0$ continue
$\phi$ and $H$ constantly from their boundary values and set
$r=r(u_a,0)+r_V(u_a,0)V$. On this ray, set $D=-r_V/H$.
The original constraint gives $D_V=r\phi_V^2/H\in L^1$ for $V<0$.
Its continuous horizon limit and constant continuation make $D$
absolutely continuous through zero. Thus $r_V=-HD$ is absolutely
continuous and satisfies $r_{VV}=\widetilde\ell_Vr_V-r\phi_V^2$
almost everywhere on the continued ray. Place the other characteristic initial
side at a late $V_j<0$. Equation \eqref{eq:extension-norm} and
the positive range of $r,H$ give a common local existence size
for these sections of this solution, by \eqref{eq:local} with
the rough direction interchanged. Taking $u_a$ close to $u_0$
and then $V_j$ close to zero gives a rectangle crossing $V=0$.
Local uniqueness identifies its $V<0$ part with the original
development.

On $V=V_j$, the $u$ constraint follows from
\eqref{eq:local-raychaudhuri} and continuity of $\Psi_u$.
The radial flux and quadratic identity in the local construction,
applied in the two null directions, propagate both constraints
throughout the rectangle:
\begin{align}\label{eq:extension-constraint}
 \frac{r_u(u,V)}{H(u,V)}
 &=\frac{r_u(u_a,V)}{H(u_a,V)}
 -\int_{u_a}^u\frac{r\phi_u^2}{H}(x,V)\dd x,\nonumber\\
 \frac{r_V(u,V)}{H(u,V)}
 &=\frac{r_V(u,V_j)}{H(u,V_j)}
 -\int_{V_j}^V\frac{r\phi_V^2}{H}(u,W)\dd W.
\end{align}
The second identity holds with $r_V/H$ absolutely continuous in $V$.
The integral equations imply the wave equations in distributions. In a regular
spherical chart, the connection coefficients are sums of bounded angular terms
and bounded smooth factors multiplying
$\widetilde\ell_u,\widetilde\ell_V,r_u,r_V$ and are locally
$L^2$. The Maxwell form $eH/(2r^2)\dd u\wedge\dd V$ satisfies
both Maxwell equations. This gives the claimed weak extension
for each solution.

\textit{Local causal identity.}
Write $V_i=V(0)$ and
\[
 \Sigma=\bigl(\{V_i\}\times[0,U)
       \;\cup\;[V_i,0)\times\{0\}\bigr)\times\mathbb S^2.
\]
These are precisely the original characteristic sides in regular
coordinates. To place them in an open spacetime, continue $a_U(v)$
by $a_U(0)$ for $v<0$, and extend the original regular fields
$(r,\phi,\widetilde\ell)$ by replacing $v,u$ with
$\max(v,0),\max(u,0)$. This gives a continuous positive-radius
metric on the past collar. With $V_-=V(-1)$, set
\[
 P=(V_-,0)\times(-1,U)\times\mathbb S^2,\qquad
 M=P\cup N,
\]
where $N=(V_j,b)\times(u_a,u_b)\times\mathbb S^2$ is the
crossing neighborhood constructed above. The original fields
and the crossing fields agree on $P\cap N$, so $M$ carries
one continuous Lorentzian metric
$g=-H\dd V\dd u+r^2\dd\omega^2$.
The auxiliary past collar is used only for causality; the field
equations hold on the original development and on $N$.

Define $D_M^+(\Sigma)$ by all past-inextendible locally Lipschitz
causal curves through a point, and set
$H_M^+(\Sigma)=\overline{D_M^+(\Sigma)}^{\,M}
\setminus I_M^-(D_M^+(\Sigma))$.
Inextendibility means absence of a past endpoint in $M$.
The future orientation gives $\dot V+\dot u\ge0$ and
\[
 r^2|\dot\omega|^2\le H\dot V\dot u,
 \qquad \dot V,\dot u\ge0
\]
almost everywhere. Timelike vectors have both null components
strictly positive.

Let a past causal curve start at $p\in N$ with $V_p<0$.
If it avoids $\Sigma$, continuity and coordinate monotonicity
confine it to
$K=[V_i,V_p]\times[0,u_p]\times\mathbb S^2\subset M$.
The ratio $H/r^2$ is bounded on this compact set, so for a constant
$C$ each ambient sphere coordinate satisfies
$|\dot\omega_i|\le C(\dot V+\dot u)$.
Local Lipschitz regularity gives absolute continuity on compact
parameter intervals. Thus $V$, $u$, and
$C(V+u)+\omega_i$ are bounded monotone functions and have past
limits. The sphere is closed, so the full curve has a past
endpoint in $K$, a contradiction. Every such curve therefore
meets the original initial set $\Sigma$.

For $p\in M$ with $V_p\ge0$, necessarily $p\in N$.
The radial null curve
$\gamma(s)=(V_p,u_a+(u_p-u_a)e^s,\omega_p)$, $s\le0$,
avoids $\Sigma$ and has its past limit outside $M$.
Consequently $D_M^+(\Sigma)\subset\{V<0\}$, with equality
after intersection with $N$. Strict increase of $V$ along
timelike curves and short radial timelike segments give
$N\cap I_M^-(D_M^+(\Sigma))=N\cap\{V<0\}$.
Taking the local closure yields
\begin{align}\label{eq:local-cauchy-horizon}
 N\cap D_M^+(\Sigma)&=N\cap\{V<0\},\nonumber\\
 N\cap H_M^+(\Sigma)&=N\cap\{V=0\}.
\end{align}
The argument uses only the original metric until a curve meets
$\Sigma$, so this identity is independent of the chosen past
collar metric. Thus the original characteristic initial surface
determines the local causal identity of the horizon.
\end{proof}

%% file: sections/local.tex
\section{Local characteristic construction}\label{sec:local}
We establish the finite-rectangle theory used for the common interior
and for continuation across the horizon.

For $\Psi=(r,\phi,\ell)$, the wave equations \eqref{eq:wave} have
smooth coefficients on a compact positive range of $r,e^\ell$ and
bilinear first-derivative sources. On a rectangle of side $d\le1$,
the integral equations give the difference estimate
\begin{equation}\label{eq:local}
 \mathcal D\le C\mathcal D_{\rm data}+C(d+\sqrt d)\mathcal D,
 \qquad
 \mathcal D=\|\delta\Psi\|_\infty+
 \|\delta\Psi_v\|_\infty+
 \|\delta\Psi_u\|_{L^2_uC_v}.
\end{equation}
The factors are the lengths of the $v$ integral and the
Cauchy--Schwarz bound for the $u$ integral. Choosing $d$ small
gives existence, uniqueness, and continuous dependence, as in the
characteristic construction of Gajic and Luk
\cite[Sec.~10]{GajicLuk2019}. The $\Lambda$ term is a smooth
zero-order potential. Here the incoming scalar derivative is in
$L^p$, and $\ell_u=-rf^2$ is in $L^{p/2}\subset L^2$.
The scalar integral equations preserve $L^p_uC_v$; the conformal
equation preserves $L^{p/2}_uC_v$. The local existence size depends
on these norms and the positive range of the metric coefficients.

The incoming constraint propagates directly at this regularity.
The radial wave equation gives
\[
 r\nu=-(r_\bullet-u)+\int_0^v re^\ell P\dd w.
\]
The map $u\mapsto\Psi(u,\cdot)$ is absolutely continuous into
$C_v$. On a compact positive radius range, smooth composition and
the bounded $v$-integration operator therefore make $\nu$ absolutely
continuous in $u$. Differentiating the displayed identity gives
$\nu_u$ almost everywhere. For
$E_u=\nu_u-\ell_u\nu+r\phi_u^2$, the wave equations and
$P'-P/r-Q=0$ give $\partial_v(rE_u)=0$.
Its incoming value vanishes. The integrating factor $e^{-\ell}$
and the horizon normalization now yield
\begin{equation}\label{eq:local-raychaudhuri}
 \frac{\nu(u,v)}{\Omega^2(u,v)}
 =-\frac14-\int_0^u\frac{r\phi_u^2}{\Omega^2}(x,v)\dd x.
\end{equation}
In particular $\nu<0$ and $0<\kappa\le1$ throughout the rectangle.
For the outgoing constraint, use
$[r\lambda]_0^u=\int_0^u re^\ell P\dd x$ and the quadratic identity
\[
 \partial_u(\lambda^2+r\lambda\ell_v-r^2\phi_v^2)
 =\partial_v(re^\ell P).
\]
The left side is the derivative of an absolutely continuous
function of $u$, since both $\Psi$ and $\Psi_v$ are $L^2_uC_v$
primitives. Integration in $u$ commutes with $v$ differentiation
in $C_u$. The outgoing data satisfy
$\lambda_{H,v}=2K_H\lambda_H-h^2/r_H$, by their radius and mass
equations. Differentiating the radial flux identity therefore gives
$\lambda_v=\ell_v\lambda-r\phi_v^2$ throughout the rectangle, and
\[
 \frac{\lambda(u,v)}{\Omega^2(u,v)}
 =\frac{\lambda(u,0)}{\Omega^2(u,0)}
 -\int_0^v\frac{r\phi_v^2}{\Omega^2}(u,w)\dd w.
\]
On a finite rectangle, $\nu<0$ gives $r\le r_H(v)$, while
$\Omega^2\le-4\nu$ yields
\[
 \int_0^u\Omega^2(x,v)\dd x\le4\bigl(r_H(v)-r(u,v)\bigr).
\]
If $r$ is bounded away from zero, the horizontal radial flux
therefore bounds $\lambda$ uniformly in terms of the fixed horizon
data and radius range. Integrating the radial equation for
$\log(-\nu)$ gives
\begin{equation}\label{eq:finite-nu}
 \nu(u,v)=-\exp\int_0^v
 \left[\kappa\left(\frac1r-\frac{e^2}{r^3}-\Lambda r\right)
                         -\frac\lambda r\right]\dd w
\end{equation}
and bounds $-\nu$ above and below. On $[0,U]\times[0,V]$, set
$F(u)=\sup_{0\le v\le V}|\phi_v(u,v)|$, which is continuous.
The vertical scalar integral gives
$|\phi_u(u,v)|\le C|f(u)|+CF(u)$ for almost every $u$ and all $v$.
Substitution into the horizontal integral yields
\[
 F(u)\le C(1+\|f\|_{L^1})+C\int_0^u F(x)\dd x.
\]
Applying Gronwall to the continuously differentiable right-hand side
bounds $F$ uniformly and gives
$\|\phi_u\|_{L^p_uC_v}\le C(1+\|f\|_{L^p})$ on bounded data sets.
The function $\log(-\nu)-\ell$ is absolutely continuous in $u$;
its almost-everywhere derivative $r\phi_u^2/(-\nu)$ integrates to
$-\log\kappa$ by the horizon normalization. The two-sided $\nu$
bounds therefore give positive lower bounds on $\kappa$ and
$\Omega^2=-4\kappa\nu$.
None of the three wave sources in \eqref{eq:wave} contains
$\ell_u$ or $\ell_v$. The preceding bounds therefore control each
source by $C(1+|f(u)|)$ for almost every $u$ and all $v$.
Integration along the characteristics gives
\[
 |\ell_u(u,v)|\le r_\bullet |f(u)|^2+C(1+|f(u)|),
 \qquad |\ell_v(u,v)|\le C.
\]
A further $v$ integration bounds $\Psi$ itself. Thus all three
norms in \eqref{eq:local} are bounded, with $\ell_u$ retaining its
$L^{p/2}_uC_v$ regularity. Restarting the characteristic integral
equations gives the finite-rectangle continuation criterion.
All these constants are uniform for the data family on each fixed
finite rectangle.